\documentclass[twoside,11pt]{article}
\usepackage{amsmath, amsthm, amsfonts}
\usepackage[figuresright]{rotating}
\usepackage{algorithm}
\usepackage[utf8]{inputenc}
\usepackage[english]{babel}

\usepackage{natbib}
\setcitestyle{authoryear,open={(},close={)}}

\def \v{\varepsilon}

\def \b{\mathbf}
\newcommand{\Y}{\b{Y}}
\newcommand{\X}{\b{X}}
\newcommand{\x}{\b{x}}
\newcommand{\y}{\b{y}}

\newtheorem{lemma}{Lemma}

\begin{document}
\renewcommand{\baselinestretch}{1}
\def\Quote{\begin{quotation}\normalfont\small}
\def\EndQuote{\end{quotation}\rm}
\def\BigHeading{\bfseries\Large}\def\MediumHeading{\bfseries\large}
\def\bct{\begin{center}}
\def\ect{\end{center}}
\font\BigCaps=cmcsc9 scaled \magstep 1
\font\BigSlant=cmsl10    scaled \magstep 1
\def\lbk{\linebreak}
\def\Report{Precision Therapeutic Biomarker Identification}
\def\Author{Liu and Rao}
\pagestyle{myheadings}
\markboth{\Author}{\Report}
\thispagestyle{empty}
\bct{\BigHeading
Precision Therapeutic Biomarker Identification with Application to the Cancer Genome Project}\\\vskip15pt
Hongmei Liu and J. Sunil Rao
\vskip5pt
{\small \it Division of Biostatistics, University of Miami, Miami, FL 33136\\
h.liu7@med.miami.edu\\
JRao@biostat.med.miami.edu}
\ect

\Quote
\vskip-5pt\noindent
Cancer cell lines have frequently been used to link drug sensitivity and resistance with genomic profiles. To capture genomic complexity in cancer, the Cancer Genome Project (CGP) (Garnett et al., 2012) screened 639 human tumor cell lines with 130 drugs ranging from known chemotherapeutic agents to experimental compounds.  Questions of interest  include: i) can cancer-specific therapeutic biomarkers be detected, ii) can drug resistance patterns be identified along with predictive strategies to circumvent resistance using alternate drugs, iii) can biomarkers of drug synergies be predicted ? To tackle these questions, following statistical challenges still exist: i)biomarkers cluster among the cell lines; ii) clusters can overlap (e.g. a cell line may belong to multiple clusters); iii) drugs should be modeled jointly.  We introduce a multivariate regression model with a latent overlapping cluster indicator variable to address above issues.   A generalized finite mixture of multivariate regression (FMMR) model in connection with the new model and a new EM algorithm for fitting are proposed.   Re-analysis of the dataset sheds new light on the therapeutic inter-relationships between cancers as well existing and novel drug behaviors for the treatment and management of cancer.
\EndQuote

\section{Introduction}
\subsection{Data Description}
\label{idd}

The use of drugs to selectively target specific genetic alterations in defined patient subpopulations has seen significant successes. One example can be found in the treatment of chronic myeloid leukaemia (CML) where the first consistent chromosomal abnormality associated with a human cancer was identified back in the 1960s.  Fast forward to the 1980s where the consequence of this abnormality was discovered to be the production  of an abnormal gene called BCR-ABL.  Intense drug discovery programs were initiated to shut down the activity of BCR-ABL, and in 1992, imatinib (Gleevec) was developed.  In 1998, the drug was tested in CML patients who had exhausted standard treatment options and whose life expectancy was limited, with remarkable results in their blood counts returning to normal.  In 2001, the FDA approved imatinib.  Today, a once commonly fatal cancer now has a five-year survival rate of 95\% \citep{druker2006five}.

Achievements like this largely inspire today's high throughput screening studies of linking cancer drugs (known or in development) to specific genomic changes which could be used as therapeutic biomarkers.  The hope is that such analyses will shed light on biological mechanisms underlying drug sensitivity, tumor resistance and potential drug combination synergies.

Cancer cell lines have frequently been used as a convenient way of conducting such studies. For a systematic search of therapeutic biomarkers to a variety of cancer drugs, the Cancer Genome Project (CGP) \citep{garnett2012systematic} screened 639 human tumor cell lines,  which represent much of the tissue-type and genetic diversity of human cancers, with 130 drugs.  These drugs, including approved drugs, drugs in development as well as experimental tool compounds, cover a wide range of targets and processes involved in cancer biology. A range of 275--507 cell lines were screened for each drug. The effect of a 72h drug treatment on cell viability was examined to derive such measures of drug sensitivity as the half-maximal inhibitory concentration ($IC_{50}$).  The cell lines underwent sequencing of 64 known cancer genes, genome-wide analysis of copy number gains and losses, and expression profiling of 14,500 genes.

Given the degree of complexity of this dataset, the multivariate analysis of variance (MANOVA) and the Elastic-Net regression applied in \citet{garnett2012systematic} are insufficient for precise knowledge discovery.  First, the marginal drug-feature associations discovered in MANOVA rarely reflect true relationships, as it is more likely that sensitivity of cancer cells to drugs depends on a multiplicity of genomic and epigenomic features with potential interactions.  Second, the Elastic-Net regression fails to concern following issues: 1) since the 639 cell lines come from a variety of cancer tissue types, there is likely additional heterogeneity manifested as subpopulations with overlaps in data; 2) note that the 130 drugs (response variables) are hardly independent, one can improve the prediction accuracy by modeling with multiple drugs \citep{breiman1997predicting}.

Moreover, there are some direct questions of interest from a subject matter perspective that we want to address.
These include, i) can cancer-specific therapeutic biomarkers be detected, ii) can drug resistance patterns be identified along with predictive strategies to circumvent resistance using alternate drugs, iii) can biomarkers of combination therapies be identified to help predict synergies in drug activities ?
To tackle these questions and previously discussed statistical challenges, we propose a multivariate regression model with a latent overlapping cluster indicator variable.  Fitting procedures inducing concurrent variable selection are introduced.

The rest of the paper is organized as follows.  In Section~\ref{lr}, we give a selective overview of existing clustering and overlapping clustering methods for general (without response variables) and regression data.  In Section~\ref{prm},  a new statistical model is introduced,  a generalized FMMR model in connection with the new model  and a new EM algorithm for fitting are provided.  We also establish a type of consistency optimality for estimation of the generalized FMMR model and perform some small simulation studies to empirically demonstrate this.  In Section~\ref{pm},  we put forward another fitting solution to the new model for comparison with the generalized FMMR model.  Section~\ref{dsd} contains a comprehensive re-analyses of the CGP data using the proposed method. Discussion is included in Section~\ref{dis2}.

\subsection{Relevant Statistical Literature Review}
\label{lr}

Clustering is a well established technique to group data elements based on a measure of similarity.  Traditional clustering techniques generate partitions so that each data point belongs to one and only one cluster. It has long been recognized that such ideal partition seldom exists in real data \citep{needham1965computer}. It is more likely that clusters overlap in some parts.  To handle the overlapping issue, \citet{lazzeroni2002plaid} put forward the well-known plaid model for two-sided overlapping clustering for gene expression data, see also \citet{turner2005improved} for improved plaid model and \citet{zhang2010bayesian} for Bayesian plaid model formulation.  Its numerical solution, however, produces unsatisfactory cluster retrievals (see Section~\ref{pm}).  Other overlapping clustering methods include the ``naive" finite mixture (FM) model with a hard threshold on posterior membership probabilities, the probabilistic model \citep{banerjee2005model} and the multiplicative mixture model based approach \citep{fu2008multiplicative}.  


More often interest centers on investigating the relationship between response variables $\Y \in \mathbb{R}^q $ and covariates $\X \in \mathbb{R}^p$ by fitting a regression model rather than to explore the $\X$ or $\Y$ on its own.  In regression analyses with $q=1$, finite mixture of regression (FMR) models are commonly used to capture unobserved cross-sectional heterogeneity in the data \citep{jedidi1996estimating}.  The FMR model postulates that a sample of observations come from a finite mixture of latent partitioning sub-populations with each sub-population represented by a regression model.  It was first introduced to statistical literatures by \citet{quandt1972new}. \citet{desarbo1988maximum} proposed an EM algorithm \citep{dempster1977maximum} based maximum likelihood estimation for the FMR model.  \citet{khalili2012variable} put forward a penalized likelihood approach for variable selection in FMR models.  As a natural generalization to multivariate responses case,  \citet{jones1992fitting} introduced the FMMR model. \citet{grun2008flexmix} proposed a $R$ package {\it flexmix} for fitting FMMR models, however it assumes that the response variables are independent.

\section{The Proposed Method}
\label{prm}
\subsection{Statistical modeling}
\label{stm}

For a sample (of cell lines) of $n$ observations, denote $\b{y}_i=(y_{i1}, \dots, y_{iq})^T \in \mathbb{R}^q$ a vector of responses ($IC_{50}$ values), $\b{x}_i=(x_{i1}, \dots, x_{ip_n})^T \in \mathbb{R}^{p_n}$ a vector of predictors (genomic markers) and $\b{\v}_i=(\v_{i1}, \dots, \v_{iq})^T$ a vector of random errors for the $i$th observation. Assume $\b{\v}_i \overset{iid} {\sim} N_{q} (\b{0}, \b{\Sigma})$ for $i=1, \dots, n$. The following proposed model allows overlapping clustering for multivariate regression data, 
\begin{align}
\b{y}_i = \sum_{k=1}^K \b{B}_k^T \b{x}_i  P_{ik} + \b{\v}_i, \quad i = 1, \dots, n, 
\label{eq1}
\end{align}
where $K$ is the total number of clusters, $\b{B}_k$ is an unknown $p_n \times q$ coefficient matrix for the $k$th cluster, and $P_{ik} \in \{0, 1\}$ is 1 if observation $i$ belongs to the $k$th cluster, otherwise 0.   In traditional clustering problem, it assumes that each observation belongs to exactly one cluster, namely $\sum_{k=1}^K P_{ik} = 1$ for all $i$.     Therefore we allow $\sum_{k=1}^K P_{ik} \ge 1$ so that each observation can belong to multiple clusters.

We provide some interpretation for the clusters in model (\ref{eq1}).  A cluster $k$ contains a subset of observations for whom $P_{ik}=1$.  We postulate that not all genomic features are relevant in describing the cluster $k$, thus assume a {\it sparse} coefficient matrix $\b{B}_k$.  Since the sparse patterns can vary with $k$, each cluster is represented by a unique set of biomarkers.   For observations belong to multiple clusters, their response variables are explained by multiple sets of biomarkers,  indicative of involving in several biological processes simultaneously. 
 

\subsection{The FMMR model}
\label{fmm}

When $\sum_{k=1}^K P_{ik} = 1$,  model (\ref{eq1}) can be characterized by a hierarchical structure which ends up with an FMMR model.   Consider a latent cluster membership random variable $z_i$ for observation $i$ from (\ref{eq1}).  Given $\sum_{k=1}^K P_{ik} = 1$, the range of $z_i$ equals to $\{1, \dots, K\}$.   Assume that $P(z_i=k) = \pi_k$ for each $k$, then $\sum_{k=1}^K \pi_k = 1$ and $\pi_k \ge 0$.  Response $\y_i$ from model (\ref{eq1}) satisfies  
\begin{align}
(\y_i \mid z_i=k, \x_i, \b{B}_k, \b{\Sigma}) \sim N_{q}(\b{B}_k^T \x_i, \b{\Sigma}).
\label{eq2}
\end{align}

Denote $\b{\Theta}=(\b{B}_1, \dots, \b{B}_K, \b{\Sigma}, \b{\pi})$ with $\b{\pi}=(\pi_1, \dots, \pi_{K-1})^T$. We can derive the joint density of $\y_i$ and $z_i$ as 
\begin{align}
f(\y_i, z_i \mid \x_i, \b{\Theta}) =  \prod_{k=1}^K \left \{ \pi_k f(\b{y}_i \mid z_i=k, \b{x}_i, \b{B}_k, \b{\Sigma}) \right \}^{I(z_i=k)}.
\label{eq3}
\end{align}

Summarizing (\ref{eq3}) over $z_i = 1, \dots, K$ yields the FMMR model
\begin{align}
f(\b{y}_i \mid \b{x}_i, \b{\Theta}) = \sum_{k=1}^K \pi_k f(\b{y}_i \mid z_i=k, \b{x}_i, \b{B}_k, \b{\Sigma}).
\label{eq4}
\end{align}

When using the EM algorithm to estimate model (\ref{eq4}),  $(\x_i, \y_i)$ is regarded as incomplete data for missing $z_i$,  so one works on the complete joint density in (\ref{eq3}).  For a sample of $n$ observations from (\ref{eq1}), the complete log-likelihood function of $\b{\Theta}$ becomes
\begin{align}
l_n(\b{\Theta}) = \sum_{i=1}^n \sum_{k=1}^K z_{ik} \log \pi_k + \sum_{i=1}^n \sum_{k=1}^K z_{ik} \log f(\b{y}_i \mid z_i=k, \b{x}_i, \b{B}_k, \b{\Sigma}), 
\label{eq5}
\end{align}
where $z_{ik} = I(z_i=k)$.

Since we do not expect all genomic markers to be informative, variable selection, to obtain a parsimonious model is necessary.   \citet{khalili2012variable} and \citet{khalili2013regularization} (for diverging model size) introduced a penalized likelihood approach for variable selection in FMR models, which was shown consistent in variable selection and highly efficient in computation.  We define a penalized complete log-likelihood function as 
\begin{align}
\tilde{l}_n(\b{\Theta}) = l_n(\b{\Theta}) - \rho_n(\b{\Theta}), 
\label{eq6}
\end{align}
where 
$$ \rho_n(\b{\Theta}) = \sum_{k=1}^K \pi_k \rho_{nk}(\b{B}_k). $$

The following two penalty functions are currently used to meet different demands in variable selection: 

1) the $L_1$-penalty in LASSO \citep{tibshirani1996regression} for simultaneous estimation and variable selection
$$ \rho_{nk}(\b{B}_k) = \lambda_k \| \b{B}_k \|_1 $$

2) a linear combination of the $L_1$- and $L_2$-penalty in Elastic-Net \citep{zou2005regularization} for simultaneous estimation and selection of grouped features
$$ \rho_{nk}(\b{B}_k) = \lambda_{k1} \|\b{B}_k\|_1 + \lambda_{k2} \|\b{B}_k\|_2^2, $$
where $\| \cdot \|_1$ and $\| \cdot \|_2$ are respectively the $L_1$- and $L_2$-norms and the tuning parameters $\lambda_{k}, \lambda_{k1}, \lambda_{k2} \ge 0$.  The $L_1$-penalty is singular at the origin, thus can shrink some coefficients to exact 0 for sufficiently large $\lambda_{k}$ or $\lambda_{k1}$ \citep{fan2001variable}.

Their revised EM algorithm can easily be generalized to maximize $\tilde{l}_n(\b{\Theta})$ in (\ref{eq6}),  hence we bypass here. In next section,  generalized FMMR model and EM algorithm are devised to fit overlapping multivariate regression data.

\subsection{The Generalized FMMR Model}

Finite mixture models (including FMR and FMMR) can only fit partitions, as can be seen from Section \ref{fmm}.  To enable overlapping clustering, one often uses a "naive" approach by applying a hard threshold $\alpha$ to finite mixture models.  This approach assigns a sample to cluster $k$ if its posterior probability of belonging to cluster $k$ is larger than a pre-specified $\alpha$, hence enabling a sample to belong to multiple clusters.  As pointed out by \cite{banerjee2005model}, this method is problematic because it is not a natural generative model for overlapping clustering, since one underlying assumption of the finite mixture model is that each observation comes from one and only one mixture component.  In this section however, we introduce a generalized FMMR model to retrieve overlapping clusters when $\sum_{k=1}^K P_{ik} \ge 1$.  This generalized model retains its ability in fitting partitions.

\begin{figure}
\centering
\includegraphics[scale=0.25]{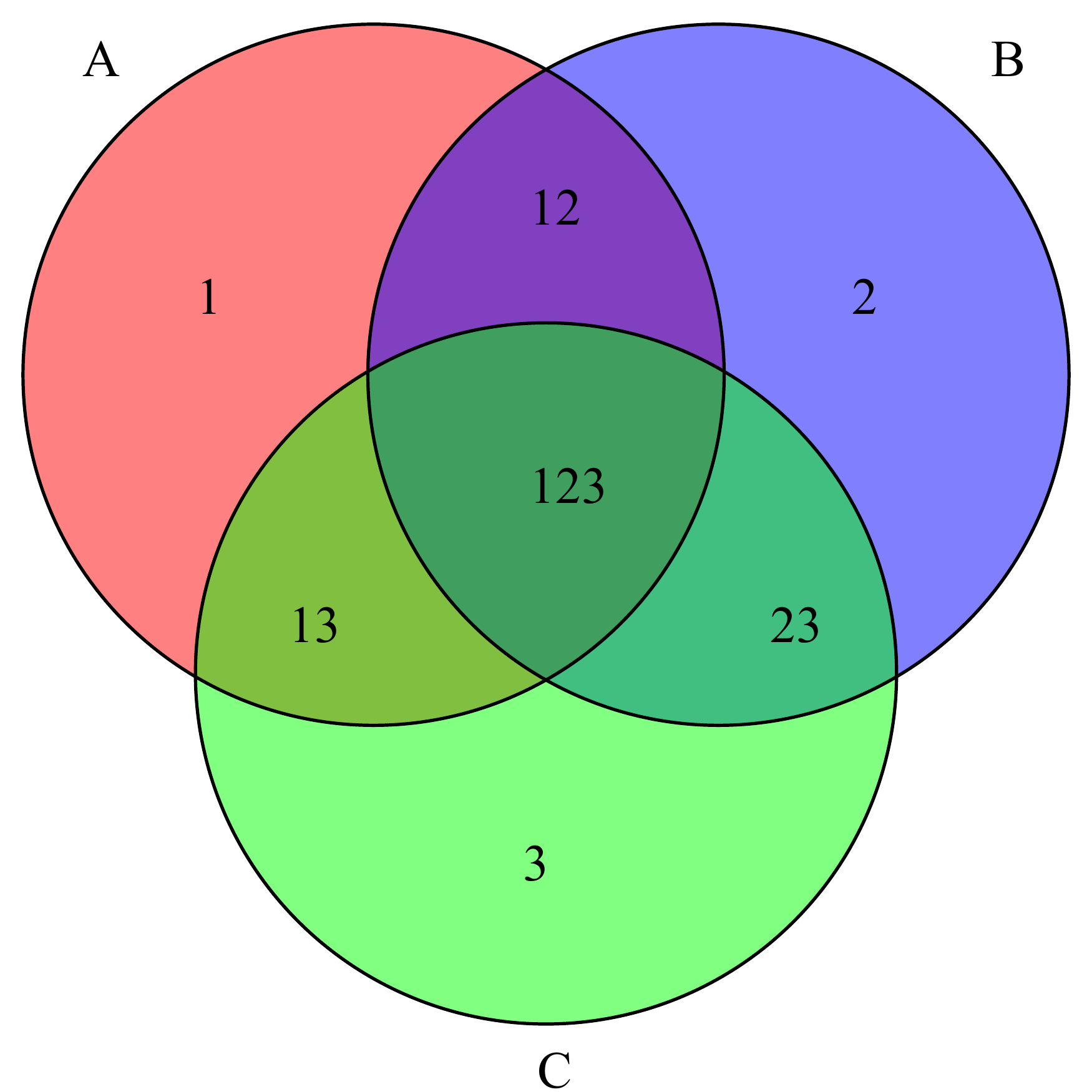}
\caption{Overall overlapping patterns given $K=3$.}
\label{fig:ol}
\end{figure}

Suppose we have $K$ {\it objective clusters} indexed by 1 to $K$.  These objective clusters refer to clusters defined in (\ref{eq1}) and can overlap with each other, resulting in $2^K-1$ types of overlapping patterns.  Let $K$ equal to 3 for an example, overall overlapping patterns are then composed of $S=\{1, 2, 3, 12, 13, 23, 123\}$.   Figure~\ref{fig:ol} shows the particular part each overlapping pattern in $S$ represents. By definition, overlapping patterns in $S$ are mutually exclusive.  And each observation $i$ from (\ref{eq1}) belongs to one and only one overlapping pattern.

In the next, we define $2^K-1$ {\it hypothetical clusters}.  Each hypothetical cluster (simply called ``cluster" since after) represents an overlapping pattern defined in advance and is indexed by an element in $T$,
$$T = \cup_{s=1}^K \left \{ (l_1 \dots l_s) : \{l_1, \dots, l_s\} \subseteq \{1, \dots, K\} \right \}. $$
Cluster $(l_1 \dots l_s)$ implies its members belonging to objective clusters $l_1, \dots, l_s$.

We now introduce a latent cluster membership random variable $z_i$ for observation $i$ from model (\ref{eq1}) and characterize the model via a hierarchical structure. Given $\sum_{k=1}^K P_{ik} \ge 1$,  the range of $z_i$ becomes $T$.  Further define 
\begin{align}
P \left (z_i = (l_1 \dots l_s) \right) = \pi_{(l_1 \dots l_s)}.
\label{eq7}
\end{align} 
Then vector $\b{\pi}=(\pi_{(l_1 \dots l_s) \in T})^T$ has
$$\sum_{(l_1 \dots l_s) \in T} \pi_{(l_1 \dots l_s)}=1, \quad \pi_{(l_1 \dots l_s)} \ge 0.$$ 
Response $\y_i$ from (\ref{eq1}) satisfies
\begin{align}
\left (\y_i \mid z_i=(l_1 \dots l_s), \x_i, \{\b{B}_{l_k}\}_{k=1}^{s}, \b{\Sigma} \right ) \sim N_q \left (\sum_{k=1}^s \b{B}_{l_k}^T \x_i, \b{\Sigma} \right).
\label{eq8}
\end{align}

Let $\b{\Theta}=(\b{B}_1, \dots, \b{B}_K, \b{\Sigma}, \b{\pi})$. By (\ref{eq7}) and (\ref{eq8}), the joint density of $\y_i$ and $z_i$ equals 
\begin{align}
f(\y_i, z_i &\mid \x_i, \b{\Theta}) = 
\label{eq9} 
\\
&\prod_{(l_1 \dots l_s) \in T} \left \{ \pi_{(l_1 \dots l_s)} f \left (\b{y}_i \mid z_i, \b{x}_i, \{\b{B}_{l_k}\}_{k=1}^s, \b{\Sigma} \right) \right \}^{I(z_i = (l_1 \dots l_s))}.  \nonumber
\end{align}

Summarizing (\ref{eq9}) over $z_i$ leads to the generalized FMMR model 
\begin{align}
f (\b{y}_i \mid \b{x}_i, \b{\Theta}) =  \sum_{(l_1 \dots l_s) \in T} \pi_{(l_1 \dots l_s)} f \left (\b{y}_i \mid z_i, \b{x}_i, \{\b{B}_{l_k}\}_{k=1}^s, \b{\Sigma} \right ). 
\label{eq10}
\end{align}
Note that if $\pi_{(l_1 \dots l_s)} = 0$ for $s > 1$, (\ref{eq10}) reduces to the traditional FMMR model in (\ref{eq4}). 

Then the (conditional) log-likelihood function of $\b{\Theta}$ for a sample of $n$ observations from (\ref{eq1}) is 
\begin{align*}
l_n^0(\b{\Theta}) = \sum_{i=1}^n \log \left(\sum_{(l_1 \dots l_s) \in T} \pi_{(l_1 \dots l_s)} f \left (\b{y}_i \mid z_i, \b{x}_i, \{\b{B}_{l_k}\}_{k=1}^s, \b{\Sigma} \right) \right).
\end{align*}

Maximizing above ordinary likelihood function yields non-zero estimates for all regression coefficients.  To induce variable selection and remove noise predictors from the regression model, 
we propose a penalized log-likelihood function
\begin{align}
\tilde{l}_n^0 (\b{\Theta}) = l_n^0(\b{\Theta}) -  \sum_{k=1}^K  \sum_{\substack{(l_1 \dots l_s): \\ k \in \{l_1, \dots, l_s\}}} \pi_{(l_1 \dots l_s)} \rho_{nk}(\b{B}_k),
\label{eq11}
\end{align}
where $\rho_{nk}(\b{B}_k)$ is the LASSO or Elastic-Net penalty function in Section~\ref{fmm}.  Notation $\sum_{\substack{(l_1 \dots l_s): \\ k \in \{l_1, \dots, l_s\}}}$ indicates that it summarizes over $(l_1 \dots l_s) \in T$ for which $k \in \{l_1, \dots ,l_s \}$.  The penalty imposed on $\b{B}_k$ is proportional to $\sum_{\substack{(l_1 \dots l_s): \\ k \in \{l_1, \dots, l_s\}}}  \pi_{(l_1 \dots l_s)}$, which by definition is proportional to the number of observations involved in the $k$th objective cluster.  This is a strategy, similar to \citet{khalili2012variable}, of relating the penalty to sample sizes for enhanced power of the method.

\subsection{Numerical Solution to the Generalized FMMR Model}
\label{nsg}

We use the renowned EM algorithm for optimization of the generalized FMMR model.  The complete log-likelihood function of $\b{\Theta}$ is
\begin{align}
l_n(\b{\Theta}) = &\sum_{i=1}^n \sum_{(l_1 \dots l_s) \in T} z_{i, (l_1 \dots l_s)} \log \pi_{(l_1 \dots l_s)}  \nonumber \\
&+ \sum_{i=1}^n \sum_{(l_1 \dots l_s) \in T} z_{i, (l_1 \dots l_s)} \log f \left (\b{y}_i \mid z_i, \b{x}_i, \{\b{B}_{l_k}\}_{k=1}^s, \b{\Sigma} \right ),  \nonumber
\end{align}
where $z_{i, (l_1 \dots l_s)} = I(z_i = (l_1 \dots l_s))$.

The penalized complete log-likelihood function for concurrent variable selection is then
\begin{align}
\tilde{l}_n(\b{\Theta}) = l_n(\b{\Theta}) -  \sum_{k=1}^K  \sum_{\substack{(l_1 \dots l_s): \\ k \in \{l_1, \dots, l_s\}}} \pi_{(l_1 \dots l_s)} \rho_{nk}(\b{B}_k).
\label{eq17}
\end{align}

Due to the overlap setting,  each $\b{B}_k$ involves in multiple clusters. Therefore coefficient matrix $\b{B}_k$ for $k=1, \dots, K$ can not be optimized independently like in usual EM algorithm. Instead we sequentially update  $\b{B}_k$ given the rest are known.  Specifically take $\rho_{nk}(\b{B}_k)$ as a LASSO penalty function for an example, our revised EM algorithm iteratively maximizes $\tilde{l}_n(\b{\Theta})$ in two steps:

E-step: Given $\b{\Theta}^m$, estimate $z_{i, (l_1 \dots l_s)}^m$ with its posterior probability by applying the Bayes' rule to (\ref{eq9}) and (\ref{eq10}), 
\begin{align}
\hat{P} (z_{i, (l_1 \dots l_s)}^m = 1 \mid \y_i, \x_i, \b{\Theta}^m) = {\pi_{(l_1 \dots l_s)}^m f \left (\b{y}_i \mid z_i, \b{x}_i, \{\b{B}_{l_k}\}_{k=1}^s, \b{\Sigma}^m \right ) \over \sum_{(l_1' \dots l_s') \in T} \pi_{(l_1' \dots l_s')}^m f \left (\b{y}_i \mid z_i, \b{x}_i, \{\b{B}_{l_k'}\}_{k=1}^s, \b{\Sigma}^m \right )}.  \nonumber
\end{align}

M-step: Given $\b{Z}^m = (z_{i, (l_1 \dots l_s)}^m)$, update the mixing proportions $\b{\pi}^{m+1}$  by 
$$ \pi_{(l_1 \dots l_s)}^{m+1} = {1 \over n} \sum_{i=1}^n z_{i, (l_1 \dots l_s)}^m, \quad (l_1 \dots l_s) \in T. $$
Note that this is obtained by maximizing the leading term of (\ref{eq17}) with respect to $\b{\pi}$.  This simplified updating scheme however works well in simulation studies.

Given $\b{Z}^m$, $\b{\pi}^{m+1}$ and $\b{\Sigma}^m$, sequentially update $\b{B}_k^{m+1} \mid \{\b{B}_s^m\}_{s=1, s \neq k}^{s=K}$ by
\begin{align}
\b{B}_k^{m+1} =  \arg \max_{\b{B}_k} \sum_{i=1}^n \sum_{\substack{(l_1 \dots l_s): \\ k \in \{l_1, \dots, l_s\}}}  z_{i, (l_1 \dots l_s)}^m \log f \left (\b{y}_i \mid z_i, \b{x}_i, \{\b{B}_{l_j}^m\}_{j=1, l_j \neq k}^s , \b{\Sigma}^m \right )  \nonumber \\
- \sum_{\substack{(l_1 \dots l_s): \\ k \in \{l_1, \dots, l_s\}}} \pi_{(l_1 \dots l_s)}^{m+1} \lambda_k \| \b{B}_k \|_1.  
\label{eq12}
\end{align}

By (\ref{eq8}),  
\begin{align}
\b{B}_k^{m+1} 
= \arg \min_{\b{B}_k} \sum_{\substack{(l_1 \dots l_s): \\ k \in \{l_1, \dots, l_s\}}} tr \big( (\b{Y}^* - \b{X} \b{B}_k )^T \b{z}_{(l_1 \dots l_s)}^m (\b{Y}^* - \b{X} \b{B}_k) (\b{\Sigma}^m)^{-1} \big) \nonumber \\
+ 2 \sum_{\substack{(l_1 \dots l_s): \\ k \in \{l_1, \dots, l_s\}}} \pi_{(l_1 \dots l_s)}^{m+1} \lambda_k \| \b{B}_k \|_1, 
\label{eq13}
\end{align}
where 
$$\b{y}_i^* = \b{y}_i - \sum_{l \in \{l_1, \dots, l_s\} \setminus k} (\b{B}_l^m)^T \b{x}_i,  \quad  \Y^* = (\y_1^*, \dots, \y_q^*),$$
$$\quad \b{z}_{(l_1 \dots l_s)}^m = diag \left (z_{1, (l_1 \dots l_s)}^m, \dots, z_{n, (l_1 \dots l_s)}^m \right ).$$ 
Above (\ref{eq13}) is a multivariate regression problem with a LASSO penalty for estimation, which can be solved by the MRCE algorithm \citep{rothman2010sparse}.  If we ignore the covariance structure of $\b{\Sigma}^m$ merely in (\ref{eq13}),  estimation of $\b{B}_k^{m+1}$ reduces to $q$ independent LASSO regression problems.   In this case, more complex penalty functions such as the Elastic-Net and fused LASSO penalties can easily be applied.  Therefore we also have investigated the performance of this simplified strategy in simulation studies.  It works surprisingly comparable to original method utilizing the MRCE algorithm for estimation of $\b{B}_k^{m+1}$.

Given $\b{Z}^m$ and $\{\b{B}_k^{m+1}\}_{k=1}^K$, update $\b{\Sigma}^{m+1}$ by
\begin{align}
\b{\Sigma}^{m+1} = &\arg \max_{\b{\Sigma}} \sum_{i=1}^n \sum_{(l_1 \dots l_s) \in T} z_{i, (l_1 \dots l_s)}^m \log f \left (\b{y}_i \mid z_i, \b{x}_i, \{\b{B}_{l_k}^{m+1}\}_{k=1}^s,  \b{\Sigma} \right ) \nonumber \\
= &\arg \min_{\b{\Sigma}} \sum_{i=1}^n \sum_{(l_1 \dots l_s) \in T} -z_{i, (l_1 \dots l_s)}^m \log|\b{\Sigma}^{-1}|  \nonumber \\
+ &\sum_{(l_1 \dots l_s) \in T} tr \big( (\b{Y} - \b{X}\sum_{k=1}^s \b{B}_{l_k}^{m+1} )^T \b{z}_{(l_1 \dots l_s)}^m (\b{Y} - \b{X} \sum_{k=1}^s \b{B}_{l_k}^{m+1}) \b{\Sigma}^{-1} \big). 
\label{eq14}
\end{align}

By taking the derivative of (\ref{eq14}) according to $\b{\Sigma}^{-1}$, we get
$$ \b{\Sigma}^{m+1} = {\sum_{(l_1 \dots l_s) \in T} (\b{Y} - \b{X}\sum_{k=1}^s \b{B}_{l_k}^{m+1})^T \b{z}_{(l_1 \dots l_s)}^m (\b{Y} - \b{X} \sum_{k=1}^s \b{B}_{l_k}^{m+1}) \over \sum_{(l_1 \dots l_s) \in T} \sum_{i=1}^n z_{i, (l_1 \dots l_s)}^m}, $$
where $\sum_{(l_1 \dots l_s) \in T} \sum_{i=1}^n z_{i, (l_1 \dots l_s)}^m = n$.

Although we can also penalize the inverse covariance matrix $\b{\Sigma}^{-1}$ of (\ref{eq14}) like in \citet{rothman2010sparse}, simulations (not presented in the paper) show that penalizing  $\b{\Sigma}^{-1}$ results in a lower degree of clustering accuracy than not penalizing  $\b{\Sigma}^{-1}$.  This may be because by penalizing  $\b{\Sigma}^{-1}$ one introduces bias into its estimation,  a biased estimate of $\b{\Sigma}^{-1}$ in return deteriorates the estimate of $z_{i, (l_1 \dots l_s)}$ in E-step. Thus we choose to not penalize  $\b{\Sigma}^{-1}$.

Commencing with an initial value $\b{\Theta}^0$, the algorithm iterates between E- and M-steps until the relative change in log-likelihood, $|\big(l_n^{m+1} (\b{\hat{\Theta}}) - l_n^{m} (\b{\hat{\Theta}}) \big) / l_n^{m} (\b{\hat{\Theta}})|$, is smaller than some threshold value, taken as $10^{-5}$ in simulation studies and $10^{-3}$ in real data analysis.  Additionally, a cluster, whose mixing proportion is smaller than some threshold value taken as 0.01 in the paper, will be removed during iterations to avoid over estimations.

\subsection{Selection of Tuning Parameters and the $K$}
\label{st}

In preceding penalized likelihood approach,  one needs to choose the values of component-wise tuning parameters $\lambda_k$ for $k=1, \dots, K$,  which controls the complexity of an estimated model.  The data-driven method cross-validation (CV) \citep{stone1977asymptotic} is frequently adopted in literatures.  Here we use a component-wise 10-fold CV method for tuning parameters selection in (\ref{eq13}). 

Moreover, selection of the number of components $K$ is essential in finite mixture models (including FMR and FMMR).  In applications, the choice can be based on prior knowledge of data analysts.  With respect to formatted methodologies, information criteria (IC) remains by far the most popular strategy for selection of $K$. See \citet{claeskens2008model} for general treatments on this topic.  Here we choose the $K$ minimizing below $IC_n$, 
\begin{align}
IC_n(K) = -2 l_n(\b{\Theta}) + N_K a_n.
\label{eq15}
\end{align}
Above $N_k$ is the effective number of parameters in the model,
$$ N_K = |\{\b{B}_k, k=1, \dots, K\}| + |\b{\pi}| -1 + |\b{\Sigma}|, $$
where $|A|$ calculates the number of {\it nonzero} elements in $A$.  In (\ref{eq15}), $a_n$ is a positive sequence depending on $n$. The well known AIC \citep{akaike1974new} and BIC \citep{schwarz1978estimating} correspond to $a_n = 2$ and $a_n = \log(n)$ respectively.  It was shown in \citet{keribin2000consistent} that under general regularity conditions, BIC can identify the true order of a finite mixture model asymptotically.  Feasibility of his results to FMR or FMMR models is unknown yet.  We examined the performance of BIC for selection of $K$ in a generalized FMMR model via simulation studies.  It obtained a high degree of accuracy.

\subsection{The Asymptotic Properties and Simulation Studies}
\label{ss}

\citet{khalili2013regularization} showed that their approach by maximizing a penalized log-likelihood function for estimation of the FMR model is consistent in both estimation and variable selection under certain regularity conditions.  Denote $f(\b{w}; \b{\Theta})$ the joint density function of data $\b{w} = (\x, \y)$ with $\b{\Theta} \in \b{\Omega}$.  The conditional density function of $\y$ given $\x$ follows an FMR model in \citet{khalili2013regularization} and a generalized FMMR model in (\ref{eq10}) in our paper.  As defined in (\ref{eq11}), the penalized log-likelihhood function is a summation of the $\log(f(\b{w}_i; \b{\Theta}))$'s minus a penalty function.  Because the regularity conditions for asymptotic establishments in \citet{khalili2013regularization} are made on $f(\b{w}; \b{\Theta})$ and the penalty function directly, their theoretical achievements can be extended to our problem.

Denote $\b{B}_{k}^j$ the $j$th column of $\b{B}_k$ for $j = 1, \dots, q$. Let $\b{B}_{k}^j = (\b{B}_{k1}^{j}, \b{B}_{k2}^{j})$ to divide the coefficient vector into non-zero and zero subsets.  Denote $\b{\Theta}^0$ the true value of $\b{\Theta}$, $\b{\Theta}^0 = (\b{\Theta}_1^0, \b{\Theta}_2^0)$ is the corresponding decomposition such that $\b{\Theta}_2^0$ contains all zero coefficients,  and $\hat{\b{\Theta}}_n = (\hat{\b{\Theta}}_{n1}, \hat{\b{\Theta}}_{n2})$ is its estimate.  Since the dimension of $\hat{\b{\Theta}}_{n1}$ increases with $n$,  we investigate the asymptotic distribution of its finite linear transformation, $\b{D}_n \hat{\b{\Theta}}_{n1}$, where $\b{D}_n$ is an $l \times d_{n1}$ constant matrix with a finite $l$ and $d_{n1}$ is the dimension of $\b{\Theta}_1^0$.  Moreover $\b{D}_n \b{D}_n^T \to \b{D}$ and $\b{D}$ is a positive definite and symmetric matrix.  We assume that $K$ is independent of the sample size $n$ and known beforehand.  Its selection is discussed in Section~\ref{st}.

\begin{lemma}
Suppose the penalty function $\rho_{nk}(\b{B}_k)$ satisfies conditions $\mathcal{P}_0 - \mathcal{P}_2$ and the joint density function $f(\b{w}; \b{\Theta})$ satisfies conditions $R_1 - R_5$ in \citet{khalili2013regularization}.  

1.  If $ {p_n^4 \over n} \to 0$, then there exists a local maximizer $\hat{\b{\Theta}}_n$ for the penalized log-likelihood  $\tilde{l}_n(\b{\Theta})$ with 
$$  \|\hat{\b{\Theta}}_n - \b{\Theta}^0\| = O_p \left \{\sqrt{{p_n \over n}} (1+q_{2n}) \right \}, $$
where $q_{2n} = \max_{kij} \left \{{|\rho_{nk}'(b_{kij}^0 )| \over  \sqrt{n}}: b_{kij}^0 \neq 0 \right \}$, $b_{kij}^0$ is an entry of $\b{B}_k^0$.

2.  If $\rho_{nk}(\b{B}_k)$ also satisfies condition $\mathcal{P}_3$ in \citet{khalili2013regularization} and ${p_n^{5} \over n} \to 0$,  for any $\sqrt{n/p_n}$-consistent maximum likelihood estimate $\hat{\b{\Theta}}_n$, as $n \to \infty$ it has: \\
i. Variable selection consistency: 
$$P(\hat{\b{B}}_{k2}^j = \b{0}) \to 1, \quad k = 1, \dots, K \ \text{and} \ j=1, \dots, q.$$
ii. Asymptotic normality:
\begin{align}
\sqrt{n} \b{D}_n \mathcal{I}_1^{-1/2} (\b{\Theta}_1^0) \Bigg \{ \Bigg [ \mathcal{I}_1 (\b{\Theta}_1^0) - {\rho_n'' (\b{\Theta}_1^0)  \over n} \Bigg] (\hat{\b{\Theta}}_{n1} - \b{\Theta}_1^0) + {\rho_n' (\b{\Theta}_1^0)  \over n} \Bigg \}  \nonumber \\
\to^d N(0, \b{D}),
\label{eq16}
\end{align}
where $\mathcal{I}_1 (\b{\Theta}_1^0)$ is the Fisher information matrix under the true subset model. 
\label{lem1}
\end{lemma}

\begin{proof}
Write $\b{\Theta} = (\theta_1, \theta_2, \dots, \theta_{t_n})^T$, where $t_n$ is the total number of parameters in the model.  Then Lemma \ref{lem1} is a direct extension from \citet{khalili2013regularization}. 
\end{proof}

By Lemma \ref{lem1},  $\hat{\b{\Theta}}_n$ under the LASSO or Elastic Net penalty has a convergence rate $\sqrt{p_n/n}$ via appropriate choice of the tuning parameters.  However consistent estimation does not necessarily guarantee consistent variable selection.   By the Lemma, the LASSO or Elastic Net penalty does not lead to consistent variable selection, because on one hand $\lambda_{k}$ must be large enough to achieve sparsity, on the other hand the bias term $q_{2n}$ is proportional to the $\lambda_{k}/\sqrt{n}$.  This problem can be solved by using adaptive LASSO or adaptive Elastic Net penalty instead, which leads to concurrent  estimation  and variable selection consistency. 

We conducted a sequence of small simulation studies to numerically demonstrate optimality
properties of the new method. They are presented in Web Appendix A.

\section{Generalization of the Plaid Model}
\label{pm}

\citet{lazzeroni2002plaid} proposed the plaid model to decompose a gene expression data as the sum of overlapping layers (clusters).  Each layer contains a subset of genes and samples,  while each gene and sample can participate in multiple layers.  Denote $x_{ij}$ a gene expression data entry and $K$ is the total number of layers. The plaid model is
\begin{align*}
x_{ij} = \omega_0 + \sum_{k=1}^K (\omega_k+\alpha_{ik}+\beta_{jk})  P_{ik}Q_{jk},
\end{align*}
where $\omega_0, \omega_k, \alpha_{ik}, \beta_{jk} \in \mathbb{R}$, $P_{ik}$ is 1 if gene $i$ is in the $k$th gene-bolck and otherwise 0,  $Q_{jk}$ is 1 if sample $j$ is in the $k$th sample-block and otherwise 0.  An iterative algorithm was put forward to consecutively search for the layers.

In fact, model (\ref{eq1}) is a generalized plaid model by introducing covariates into the model so that it can cluster overlapping multivariate regression data.  Therefore, estimation procedures of the plaid model can also be used to estimate (\ref{eq1}).  In parallel to the plaid model, the parameters are estimated by minimizing the $Q$, 
$$ Q = {1 \over 2} \sum_{i=1}^n \|\b{y}_i - \sum_{k=1}^K \b{x}_i^T \b{B}_k P_{ik}\|_2^2 + \sum_{k=1}^K \rho_{nk} (\b{B}_{k}), $$
where $\rho_{nk}(\b{B}_{k})$ is the penalty function defined in Section~\ref{fmm}.

Detailed procedures for optimizing $Q$ are shown in Algorithm 1 in Web Appendix B, which is a generalization of the improved iterative algorithm by \citet{turner2005improved} to multivariate regression.  At each time, the algorithm searches for a layer that explains as much of the rest data as possible; once a layer has been found, it is subtracted from the data and remains unchanged; the algorithm stops when no further layers can be found.  However this algorithm introduces bias to parameter estimations due to its discrete updating scheme, where previous recruited layers can not be refined as new layers are recruited into the model.  To address this issue, we further revised Algorithm 1 to enable joint optimization of all layers. Detailed procedures are presented in Algorithm 2 in Web Appendix B.

\section{Therapeutic Biomarker Identification for the CGP Data}
\label{dsd}

We now turn our attention back to analyses of the CGP high throughput drug sensitivity dataset for cancer described in Section~\ref{idd}. An updated version of this data \\
(http://www.cancerrxgene.org/downloads/) contains 707 human tumor cell lines as samples, 140 drugs as response variables, and 13831 genomic features as covariates (including the tissue type, rearrangements, mutation status of 71 cancer genes,  continuous copy number data of 426 genes causally implicated in cancer, as well as genome-wide transcriptional profiles). The response variables, made up with $IC_{50}$ values from pairwise drug-cell-line screening, have some missing values. Therefore the data was first filtered by removing cell lines for which less than 50\% of the drugs were tested, resulting in 591 cell lines remaining.  In the remaining data, about 37.4\% of the cell lines are with missing values.  These missing values were then imputed via the a random forest imputation algorithm \citep{ishwaran2008random}. \\

\noindent
{\bf Direct Q1: Can cancer-specific therapeutic biomarkers be detected?} \\

We applied our new method to identify patterns of cancer-specific therapeutic biomarkers.  Note that by ``cancer-specific" we do not assume separate patterns for each cancer type but rather clusters that may be {\it driven} by one or a small number of cancer types.   Throughout the analysis, we fixed $K$ to a value of 3, although as previously shown, a BIC model selection approach could also be used.  However, fixing the value of $K$ does not limit the interesting findings that we can still find and saves on computational time.

Due to the scale and complexity of the data, the analysis was conducted in three steps.  Although simulation studies show that modeling with multiple response variables yields much higher clustering accuracy than modeling with a single response variable,  it is unreasonable to simply fit all 140 drugs (responses) in one generalized FMMR model. Because by doing so, one assumes that the cell line assignments to clusters are the same for all 140 drugs, which can hardly be true.  Thus in our analysis, we first divided the 140 drugs into several groups and then fitted each drug group with a generalized FMMR model, in which to capture active grouped genomic features, the Elastic-Net penalty along with a simplified updating scheme of $\b{B}_k$ discussed in Section \ref{nsg} was used.

In the first step, each drug $c$ was fitted by a generalized FMMR model, from which we got drug-specific cell line assignments, $\hat{\b{Z}}_c$, and cluster-wise coefficient estimates, $\hat{\b{B}}_k^c$ for $k=1, \dots, 3$.  In the second step, we used the affinity-propagation clustering (APC) algorithm \citep{frey2007clustering} to group the 140 drugs based on results from step 1. The grouping was conducted in a two-level nested manner.  In the first level, the APC algorithm was applied to all 140 drugs. The pair-wise similarity matrix required by the algorithm as input data was calculated from the Euclidean distance between $\hat{\b{Z}}_a$ and $\hat{\b{Z}}_b$ where $a$ and $b$ refer to two unique drugs.  In the second level, the APC algorithm was re-applied to  each first-level drug group.  Coefficient estimates $\hat{\b{B}}_k^c$ for $k=1, \dots, 3$ were used to calculate the pair-wise similarities.  Consequently, drugs having similar cell line assignments and cluster-wise coefficient estimates were grouped together. Resulting drug groups are shown in Web Figure 6, where the coloring and thickness of connecting lines indicate degree of closeness.   In the third step, each second-level drug group $C$ from step 2 was fitted by a generalized FMMR model, from which we got drug-group-specific cell line assignments, $\hat{\b{Z}}_C$, and cluster-wise coefficient estimates $\hat{\b{B}}_k^C$ for $k=1, \dots, 3$.

As an illustration, cell line members and coefficient estimates of clusters 1 and 12 for each second-level drug group are shown in Web Figures 7--8.  Cell line members of a cluster are depicted via frequencies of the cancer types cell lines belong to.  The frequencies are represented by the size of colored bubbles in top panel of the figures.   Estimates of other clusters are not presented in the paper due to limited space.  It's very clear that clusters being driven by different cancer types have very different therapeutic genomic profiles and yet no clusters are homogeneous in a particular cancer type.  This speaks to the great heterogeneity seen in cancer overall where in fact, cancer is not thought of as a single disease but rather many diseases characterized by different underlying biological changes.   \\

\begin{figure}
\centering
\includegraphics[scale=0.32]{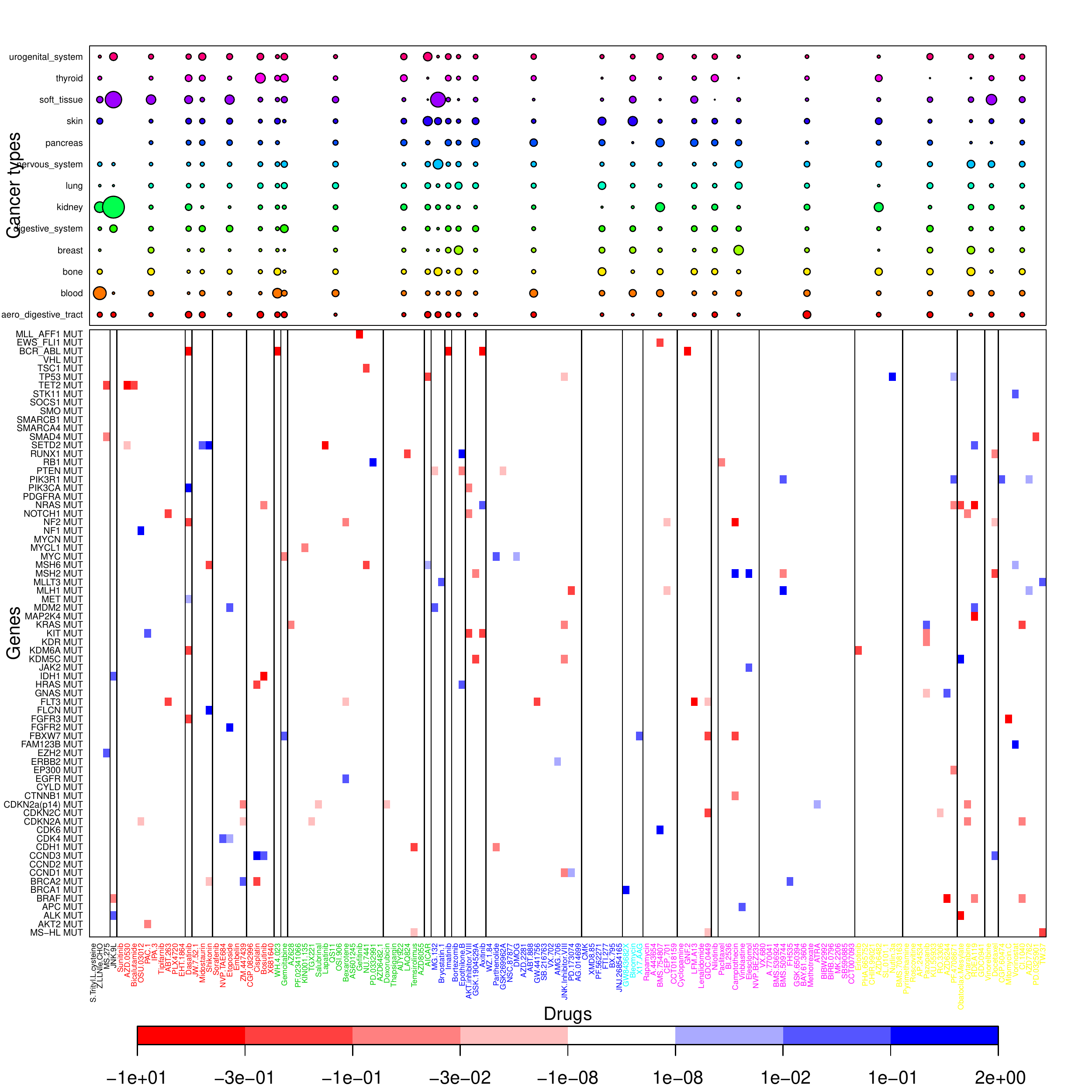}
\caption{Subset estimation results of cluster 1 by fitting the data of each second-level drug group in a generalized FMMR model.  Bottom panel shows a subset of the coefficient estimates corresponding to the mutations of 71 cancer genes, and upper panel shows cell line compositions of cluster 1 by cancer types.}
\label{fig:m1}
\end{figure}

\noindent
{\bf Direct Q2: Can drug resistance patterns be identified along with predictive strategies to
circumvent resistance using alternative drugs?}\\

We discussed in the Introduction to the paper the success story of Gleevec which was 
used for the treatment of CML based on the known specificity of targeting the BCR-ABL gene. 
It is now becoming more and more common in the treatment of cancer to first sequence known
cancer genes in tumor genomes and then design therapies accordingly \citep{bailey2014implementation}.  It is also true that
tumors can develop resistance to first line therapies by accumulating mutations which confer
resistance.  We now show how our methodology can be used to identify predictive strategies for circumventing drug resistance based on mutation data.

\begin{figure}
\centering
\includegraphics[scale=0.32]{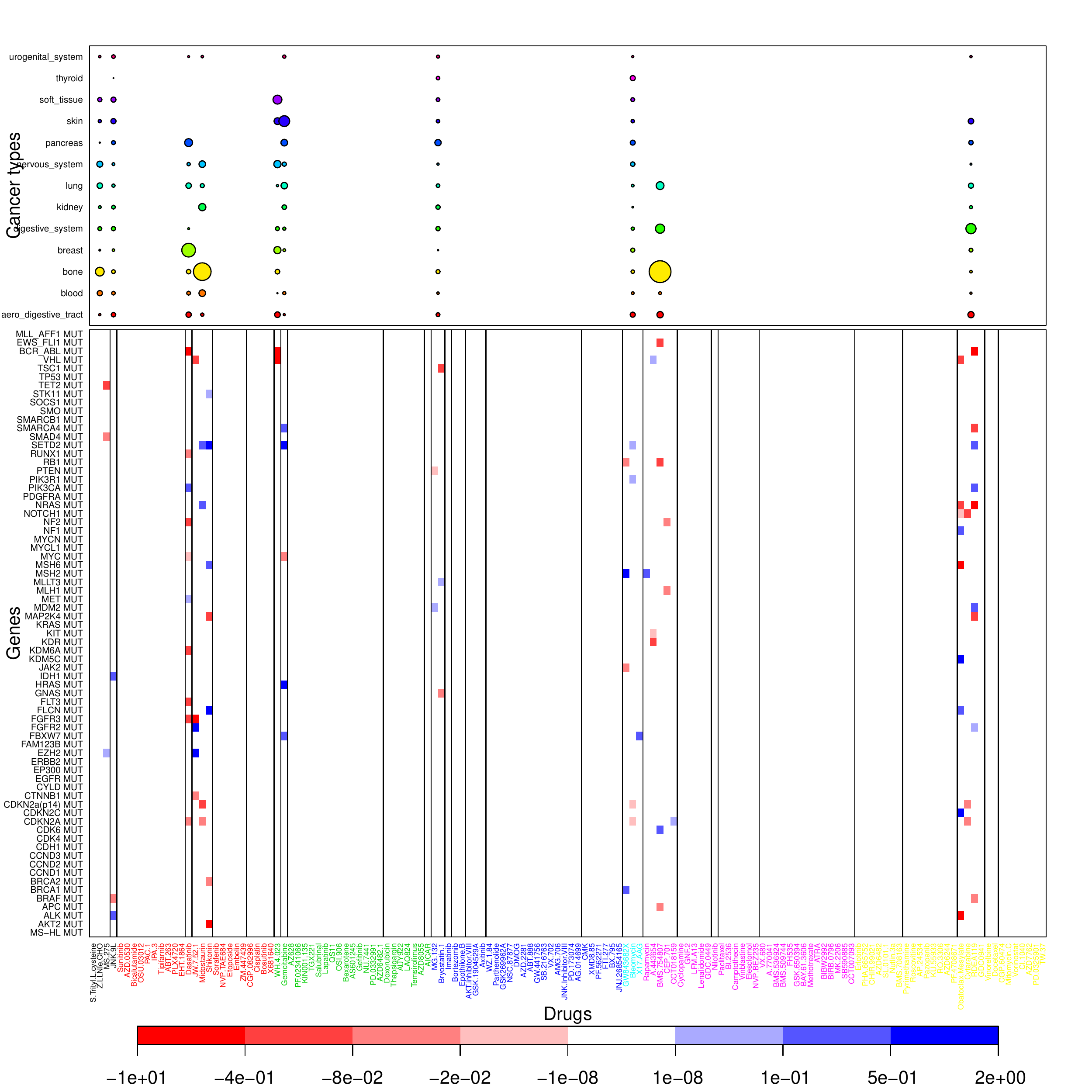}
\caption{Subset estimation results of cluster 12 by fitting the data of each second-level drug group in a generalized FMMR model.}
\label{fig:m4}
\end{figure}

Figures \ref{fig:m1}--\ref{fig:m4} are similar to Web Figures 7--8, except that they only show a subset of $\hat{\b{B}}_1^C$ and $\hat{\b{B}}_{12}^C$ with regard to the mutations of 71 cancer genes, where $\hat{\b{B}}_{12}^C = \hat{\b{B}}_{1}^C + \hat{\b{B}}_{2}^C$ are coefficient estimates of cluster 12 for drug group $C$.  Bottom panel of figure~\ref{fig:m1}--\ref{fig:m4} are in fact regression based drug-genetic association map with red indicating drug-sensitivity biomarkers and blue indicating drug-resistance biomarkers, conditioning on a cluster of cell lines (upper panel) identified by the generalized FMMR model.  We propose that above estimation results can be used for drug repurposing \citep{martins2015linking} for resistant tumors.  Take soft tissue cancers as an example of this drug repurposing.  We extracted those clusters which contain a high percentage of soft tissue cancers (large purple bubbles),  along with the subset cluster-wise coefficient estimates.  The extracted information is presented in Web Figure 5.  Its upper panel shows the $IC_{50}$ values of soft tissue cell lines in extracted clusters.   The drugs in Web Figure 5 were further filtered by only keeping those that have low $IC_{50}$ values without resistance biomarkers and that have high $IC_{50}$ values with resistance biomarkers, leading to Figure~\ref{fig:st2}.  In this figure, we conclude that the former set of drugs can be used as alternatives for the later set of drugs which progress resistance to soft tissue cancers.  
As a positive control of this strategy, drug Gemcitabine was found to be effective in soft tissue cancers resistant to standard chemotherapy (doxyrubicin) in \citet{merimsky2000gemcitabine}. \\

\begin{figure}
\centering
\includegraphics[scale=0.45]{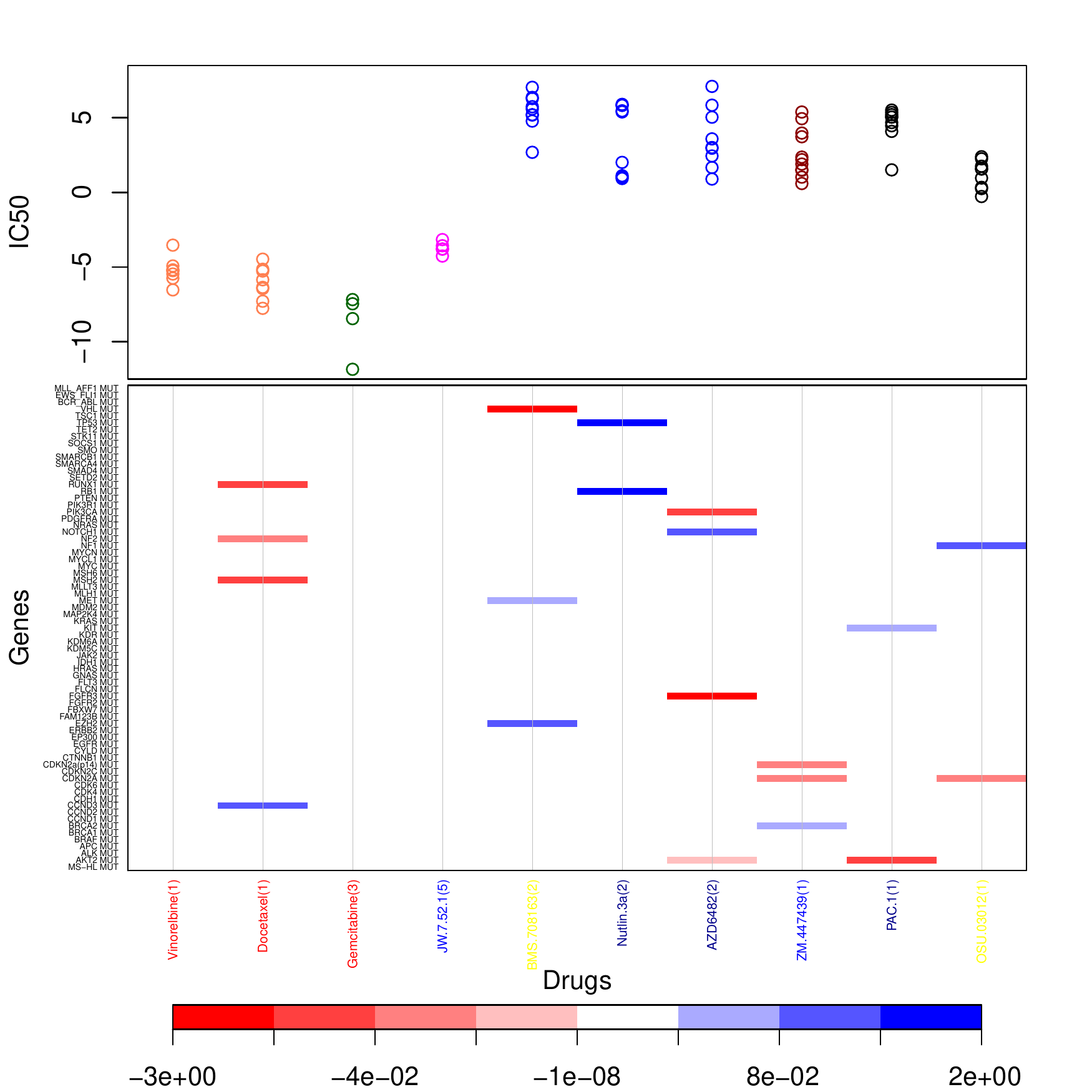}
\caption{The extracted $IC_{50}$ values of soft tissue cancers (upper panel) and a subset of the  coefficient estimates corresponding to the mutations of 71 cancer genes (bottom panel) after filtering those drugs not showing the desired pattern.}
\label{fig:st2}
\end{figure}

\noindent
{\bf Direct Q3: Can drug synergies be predicted?}\\

Our desire to allow overlapping cell line clusters becomes apparent in answering this question.  
Overlapping clusters can be used to guide drug combinations and identify potential drug synergies.  From $\hat{\b{Z}}_C$, we can get the cluster-wise $IC_{50}$ values for each drug.  Web Figure 9 are boxplots of the cluster-wise $IC_{50}$ values for 15 out of 140 drugs.  Due to limited space, the rest are not shown.  We then identify the drugs for which the magnitudes of $IC_{50}$s of overlapping cluster, say 12, are between those of clusters 1 and 2.  For instance, drugs MS.275 and GW843682X show such interesting pattern.  We focus on the cluster with highest $IC_{50}$s on average, say cluster 2, and use the coefficient estimates of overlapping cluster to guide the search of another drug to decrease the $IC_{50}$s of cluster 2 towards overlapping cluster.

\begin{figure}
\centering
\includegraphics[scale=0.45]{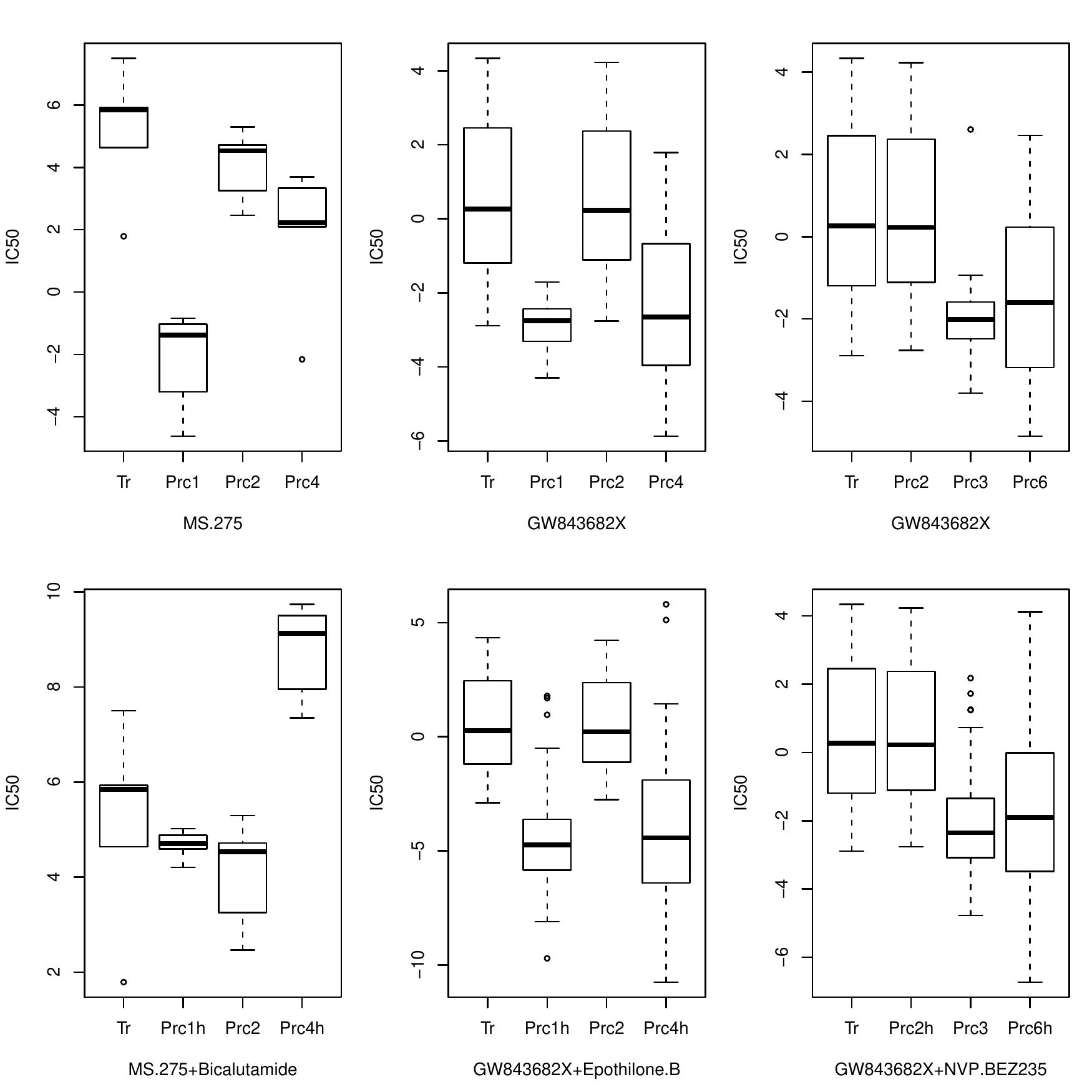}
\caption{Prediction of drug effects and drug combination effects.  In each plot, ``Tr" represents the observed $IC_{50}$ values;  the rest three boxplots are predicted $IC_{50}$  values using selected coefficient matrices.}
\label{fig:dc}
\end{figure}

Take GW843682X for an example,  cluster 2 has the highest $IC_{50}$s among clusters 1, 2, and 12.  We predicted the $IC_{50}$s in cluster 2 using $\hat{\b{B}}_{1}^{gw843682x}$, $\hat{\b{B}}_{2}^{gw843682x}$ and $\hat{\b{B}}_{12}^{gw843682x}$ respectively.  Predicted results are shown in top panel (middle) of Figure~\ref{fig:dc}.  By doing so, we want to see if coefficients $\hat{\b{B}}_{1}^{gw843682x}$ or $\hat{\b{B}}_{12}^{gw843682x}$ applied to cell lines in cluster 2 can generate low $IC_{50}$s similar to those in clusters 1 or 12.  Since $\hat{\b{B}}_{1}^{gw843682x}$ corresponds to the coefficient estimates of the target drug GW843682X,  we need to find another drug producing similar coefficient estimates to $\hat{\b{B}}_{1}^{gw843682x}$,  this leads to $\hat{\b{B}}_{3}^{epothilone.b}$.   The predicted $IC_{50}$s of  cluster 2 using $\hat{\b{B}}_{3}^{epothilone.b}$, $\hat{\b{B}}_{2}^{gw843682x}$ and  $\hat{\b{B}}_{3}^{epothilone.b} + \hat{\b{B}}_{2}^{gw843682x}$ are presented in bottom panel (middle) of Figure~\ref{fig:dc}.   It shows that the combination of GW843682X and Epothilone.B produces much lower $IC_{50}$s for cluster 2 than GW843682X alone.  Similar studies were conducted for MS.275 with respect to its clusters 1, 2, 12 and GW843682X with respect to its clusters 2, 3, 23.  Results are shown in left and right columns of Figure~\ref{fig:dc} respectively. The combination of GW843682X and NVP.BEZ235 decreases the $IC_{50}$s dramatically, while the combination of MS.275 and Bicalutamide does not show such effect.  As a positive control,  \citet{wildey2014pharmacogenomic} pointed out that a group of drugs (including GW843682X and Epothilone B) may provide a practical starting point to investigate combinatorial drug therapies for synergistic effect in small-cell lung cancer (SCLC).  Among the group of drugs,  they found synergism between GW843682X and CGP60474 via a preliminary study.

\section{Discussion}
\label{dis2}

In this paper, we proposed a new model for identifying therapeutic biomarkers for cancer which can answer specific questions regarding sensitivity, resistance and synergy.   We used a penalized likelihood approach for the FMMR model which enforced sparsity in the genomic features. To enable overlapping clustering, the FMMR model was then generalized and a new EM algorithm derived for estimation of model parameters. While the noteworthy plaid model can also be generalized for overlapping clustering multivariate regression data,  the generalized FMMR model markedly outperforms this method. 

Some improvements can be made for future developments on this work. First, other useful penalty functions for multivariate regression estimations can be adopted, such as the $MAP$ (MAster Predictor) penalty in \citet{peng2010regularized}, the $L_2SVS$ method by \citet{simila2007input} and the $L_{\infty}SVS$ method by \citet{turlach2005simultaneous}. Second, the multivariate normal distribution assumption on $\b{Y}_i$ can be extended to other more flexible parametric families of multivariate distributions, such as the skew-normal distribution \citep{azzalini2005skew} and the multivariate skew-$t$ distribution \citep{chen2014regularized}, for frequent presence of skewness and kurtosis in real data. 

As described in the CGP data analyses, we had about 37\% of the final analyses observations that had missing data and which we dealt with by using random forest imputation.  It would be of interest to explore the impact of this step more thoroughly.  An alternative would be to generalize to this problem,  the recently developed E-MS algorithm \citet{jiang2015ms} for model selection with incomplete data.



\section*{Acknowledgements}

Hongmei Liu was supported by a doctoral fellowship from 
the Sylvester Comprehensive Cancer Center at the University of Miami.\vspace*{-8pt}


\section*{Supplementary Materials}

Web Appendices, Tables, and Figures referenced in Sections \ref{ss}, \ref{pm} and \ref{dsd} are available with this paper at the Biometrics website on Wiley Online Library. \vspace*{-8pt}


%
\bibliographystyle{abbrvnat} 
\bibliography{reference}

\newpage

\bct{\BigHeading
Web-based Supplementary Materials for ``Precision Therapeutic Biomarker Identification with Application to the Cancer Genome Project"}
\ect

\appendix

\makeatletter
\renewcommand{\fnum@figure}{Web Figure \thefigure}
\renewcommand{\fnum@table}{Web Table \thetable}
\makeatother

\section{Simulation Studies}
\subsection{Design of the Simulations}

Two scenarios were designed to investigate the non-overlapping clustering and overlapping clustering of the new method respectively.  In each scenario, data were generated from model (2.1) with $K=3$, $p_n=15$, $q=3$ and $n=150, 450$.  Predictors $\b{x}_i$, $i=1, \dots, n$, were drawn independently from $N_{p_n} (\b{0}, \b{\Sigma}_1)$ with $\b{\Sigma}_1(i, j) = 0.5^{|i-j|}$.  Random errors $\b{\v}_i$, $i=1, \dots, n$, were drawn independently from $N_q(\b{0}, \b{\Sigma}_2)$ with $\b{\Sigma}_2 (i, j) = 0.75^{|i-j|}$.

The sparse coefficient matrices $\b{B}_k, k=1,2,3$, were generated as 
\begin{align}
\b{B}_k = \b{W} \otimes \b{S} \otimes \b{T}, 
\end{align}
where $\otimes$ indicates the element-wise product. Each entry of $\b{W}$ was drawn independently from $N(0, 1)$, each entry of $\b{S}$ was drawn independently from the Bernoulli distribution $B(1, p_1)$, and each row of $\b{T}$ (either all 1 or all 0) was determined by an independent draw from $B(1, p_2)$.   As a result, we expect $p_1p_2p_n$ relevant predictors for each response variable, and $(1-p_2)p_n$ predictors are expected to be irrelevant to all $q$ response variables.   We let $p_1=0.5$ and $p_2=0.9$.  Each simulation was repeated for 50 times.  A new sequence of $\b{B}_k, k=1,2,3$ and a new set of data $(\X, \Y)$ were generated for each repetition.  We assumed that $K$ is known, so that performance evaluation of the new method is not interfered by selection of $K$.  Its selection is separately examined in Section~\ref{as}. 

Scenario 1: The 3 clusters are non-overlapping, each cluster contains $n/3$ observations.  

Scenario 2: 70\% of the observations involve in single cluster, 22\% of the observations involve in two clusters and 8\% of the observations involve in three clusters.

\subsection{Quality Measures of Cluster Recovery }
\label{qm}

Quality measures in \citet{baeza1999modern} have been used to evaluate the clustering performance of the proposed method.  Denote $N_S$ the number of elements in $S$.  Below quality measures were proposed to compare a retrieved cluster $\hat{A}$ with a target cluster $A$, 
\begin{align}
&\text{``specificity"} = {N_{A \cap \hat{A}} \over N_A},  \quad  \text{``sensitivity"} = {N_{A \cap \hat{A}} \over N_{\hat{A}}}, \nonumber \\
&``F_1 \ \text{measure}" = {2N_{A \cap \hat{A}} \over N_A+N_{\hat{A}}}. \nonumber
\end{align}
The $F_1$ measure, taken as the harmonic mean of specificity and sensitivity, gives an overall measure of the clustering. 

Moreover, we employed the one-to-one correspondence match approach \citep{turner2005improved} to evaluate a sequence of retrieved clusters.  It includes about three steps.  First make the number of retrieved clusters to be the same as the number of target clusters by adding null clusters to retrieved clusters or dropping addition poorly retrieved clusters. Retrieved clusters are then matched to target clusters in pair. Finally, calculate the mean quality measures for the sequence of paired clusters.

\subsection{Results of the Simulations}
\label{rs}

A sequence of methods were implemented for comprehensive comparison with our new method,  including the generalized plaid model ({\it plaid}) and its revised counterparts ({\it aplaid});  the R package {\it flexmix} ({\it EM}) where response variables are treated as independent and overlap issues are not considered;  the generalized FMMR model with response variables treated as independent ({\it gEMseplasso0});  the generalized FMMR model with separate LASSO estimation of $\b{B}_k$ ({\it gEMseplasso});  and finally the original generalized FMMR model with MRCE estimation of $\b{B}_k$ ({\it gEMmrce}).

Their performances were evaluated via the three quality measures introduced in Section~\ref{qm} and the sum of squared errors (SSE) of coefficients estimates.  Results are summarized in Web Figure~\ref{fig:case1}--\ref{fig:case2}.  In scenario 1 when there is no overlap, {\it EM} and {\it gEMseplasso0} do equally well showing that the generalized FMMR model retains the ability to recover non-overlapping clusters. By taking into account of the covariance structure among response variables, {\it gEMseplasso} and {\it gEMmrce} outperform {\it EM} and {\it gEMseplasso0}, but there is little difference between {\it gEMseplasso} and {\it gEMmrce}.  In scenario 2 when there are 30\% overlaps,  specificity of {\it EM} drops down dramatically, resulting in a poor $F_1$ measure.  The {\it gEMseplasso0} does much better than {\it EM} by estimating the overlaps,  while  it is defeated by {\it gEMseplasso} and {\it gEMmrce}.  In this scenario, {\it gEMmrce} outperforms {\it gEMseplasso} when $n=150$.  However, they converge very fast with $n$,  both achieving a median 95\% clustering accuracy when $n=450$.  The {\it plaid} and {\it aplaid} have a poor performance in all scenarios.

\subsection{Additional Simulations}
\label{as}

We conducted another simulation for scenario 2 to evaluate the accuracy of BIC in selection of $K$.  
In the simulation, $K$ is unknown in candidate procedures {\it EM} and {\it gEMseplasso},  and is chosen by minimizing the BIC.  Results are summarized in Web Figure~\ref{fig:ksimu} and Web Table~\ref{tab:BIC}.  Given 30\% overlaps in simulated data,  BIC has a poor performance in identifying the true $K$ with the {\it EM} procedure.   Whereas by fitting the overlapping clusters with {\it gEMseplasso},  BIC obtains a much higher accuracy in selection of $K$ (see Web Table~\ref{tab:BIC}).   Moreover, although $K$ is misspecified for 12\% of the times when $n=450$,  {\it gEMseplasso} still has achieved a median 93\% clustering accuracy (see Web Figure~\ref{fig:ksimu}).

An additional simulation was conducted for scenario 2 to investigate the advantage of modeling with multivariate response variables in the new method.  We compared the performance of modeling with all 3 response variables versus modeling with each of the 3 response variables in the generalized FMMR model.  Results are summarized in Web Figure~\ref{fig:unisimu}.  Using multivariate response variables leads to a much higher clustering accuracy.

\section{Two algorithms of the Generalized Plaid Model}
\begin{algorithm}
For {$ k = 1:K $} do

At this point, we have already found $k-1$ layers and are searching for the $k$th layer.  Calculate the residuals from previous $k-1$ layers by 
$$ \b{z}_i = \b{y}_i - \sum_{l=1}^{k-1} \b{x}_i^T \b{B}_l P_{il}, \quad i=1, \dots, n.$$

Initialize $\b{P}_k^0$. 

For {$ s = 1:S $} do

Given $\b{P}_k^{s-1}$, compute 
$$ \b{B}_k^{s} = \arg \min_{\b{B}_k} {1 \over 2} \sum_{i=1}^n \|\b{z}_i - \b{x}_i^T \b{B}_k P_{ik}^{s-1}\|_2^2 + \sum_{j=1}^q \lambda_{jk} \|\b{b}_{jk}\|_1, $$
where $\b{B}_k = (\b{b}_{1k}, \dots, \b{b}_{qk}).$

Given $\b{B}_k^s$, compute 
\begin{eqnarray}
P_{ik}^s = 
\begin{cases}
0.5 + \min \{0.5, {S \over 2(S-T)} \}, \quad \|\b{z}_i - \b{x}_i^T \b{B}_k^s\|_2^2 \le \|\b{z}_i\|_2^2,  \\
0.5 - \min \{0.5, {S \over 2(S-T)} \}, \quad \text{otherwise}, \nonumber
\end{cases}
\end{eqnarray}
where $T$ was taken as $0.2S$ for $i=1, \dots, n$. 

End for

Given $\b{P}_k^S$, update $\b{B}_k^{S+1}$.

Given $\b{B}_k^{S+1}$, prune layer memberships by 
\begin{eqnarray}
P_{ik}^{S+1} = 
\begin{cases}
1, \quad \|\b{z}_i - \b{x}_i^T \b{B}_k\|_2^2 \le \tau \|\b{z}_i\|_2^2,  \\
0, \quad \text{otherwise}, \nonumber
\end{cases}
\end{eqnarray}
where $\tau \in (0, 1)$ for $i=1, \dots, n$.

Given $\b{P}_k^{S+1}$, update $\b{B}_k^{S+2}$.

Given $\b{P}_1^{S+1}$, \dots, $\b{P}_k^{S+1}$, back fit the layers $R$ times. We set $R=2$. 

End for
\caption{Plaid}
\label{plaid}
\end{algorithm}

\begin{algorithm}
Initialize $\b{B}_1^0, \dots, \b{B}_{K}^0$ and $\b{P}_1^0, \dots, \b{P}_{K}^0$. 

For {$ s = 1:S $} do

For {$ k = 1:K $} do

Given $(\b{B}_1^s, \dots, \b{B}_{k-1}^s, \b{B}_{k+1}^{s-1}, \dots, \b{B}_{K}^{k-1})$ and $(\b{P}_1^s, \dots, \b{P}_{k-1}^s, \b{P}_{k}^{s-1}, \dots, \b{P}_{K}^{s-1})$, compute
$$ \b{z}_i = \b{y}_i - \sum_{l=1}^{k-1} \b{x}_i^T \b{B}_l^s P_{il}^s - \sum_{l=k+1}^{K} \b{x}_i^T \b{B}_l^{s-1} P_{il}^{s-1}, \quad  i = 1, \dots, n, $$ 

$$ \b{B}_k^s = \arg \min_{\b{B}_k} {1 \over 2} \sum_{i=1}^n \|\b{z}_i - \b{x}_i^T \b{B}_k P_{ik}^{s-1}\|_2^2 + \sum_{j=1}^q \lambda_{jk} \|\b{b}_{jk}\|_1, $$  
where $\b{B}_k = (\b{b}_{1k}, \dots, \b{b}_{qk}).$

Given $\b{B}_k^s$, compute 
\begin{eqnarray}
P_{ik}^s = 
\begin{cases}
0.5 + \min \{0.5, {S \over 2(S-T)} \}, \quad  \|\b{z}_i - \b{x}_i^T \b{B}_k\|_2^2 \le \tau\|\b{z}_i\|_2^2,  \\
0.5 - \min \{0.5, {S \over 2(S-T)} \}, \quad  \text{otherwise},  \nonumber
\end{cases}
\end{eqnarray}
where $T$ was taken as $0.2S$ and $\tau \in (0, 1)$ for $i = 1, \dots, n$. 

End for

End for
\caption{All-plaid}
\label{allplaid}
\end{algorithm}

\section*{Web Tables}

\begin{table}
\begin{center}
\caption{Accuracy of BIC in selection of $K$. \label{tab:BIC}}
\begin{tabular}{ccc}
\hline
sample size & {\it EM} & {\it gEMseplasso} \\
\hline
n=150 & 0.32 & 0.56 \\
n=450 & 0.22 & 0.88 \\
\hline
\end{tabular}
\end{center}
\bigskip
\end{table}

\section*{Web Figures}

\begin{figure}
\centering
\includegraphics[scale=0.38]{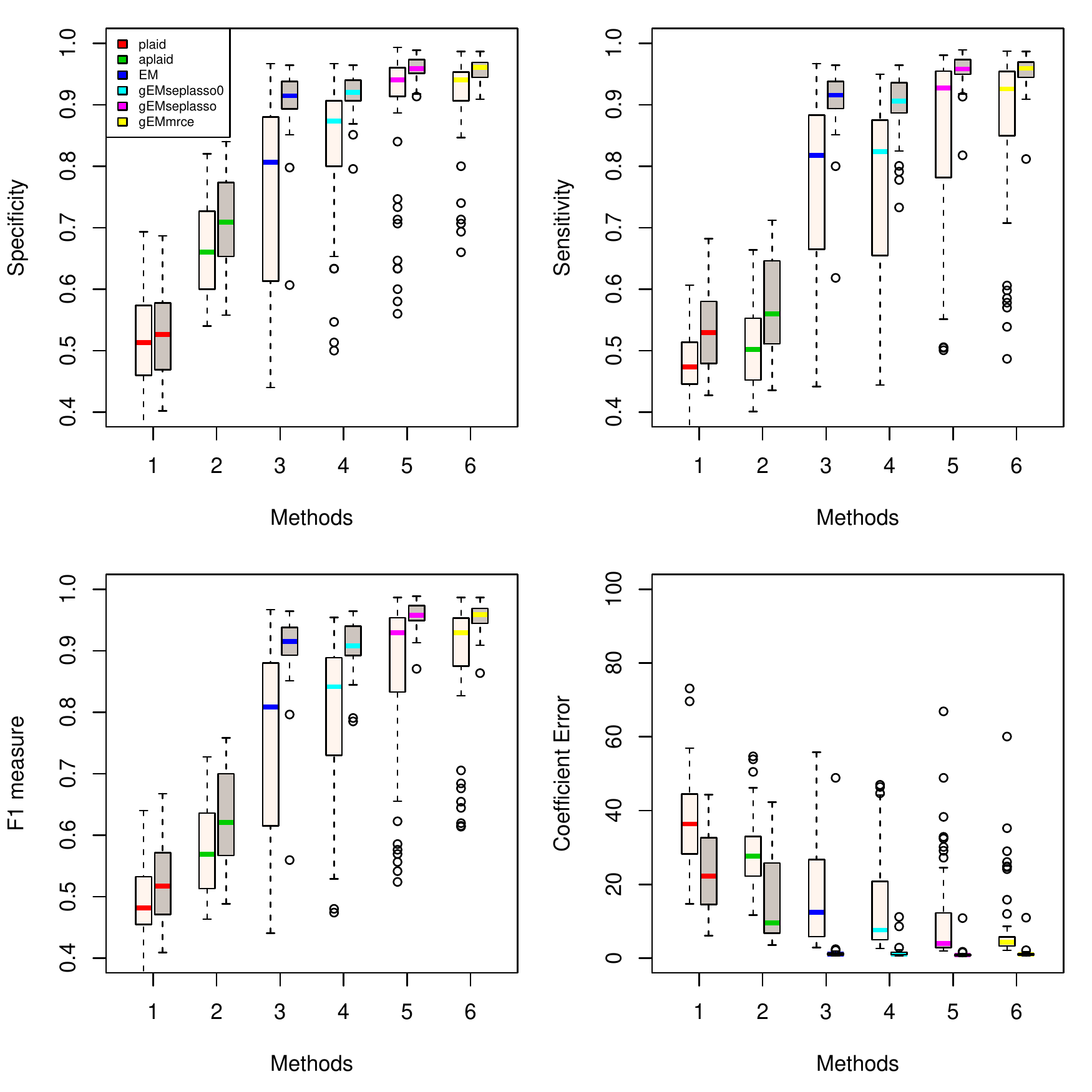}
\caption{Simulation results of scenario 1. The two box plots under each method correspond to $n=150, 450$ respectively.}
\label{fig:case1}
\end{figure}

\begin{figure}
\centering
\includegraphics[scale=0.38]{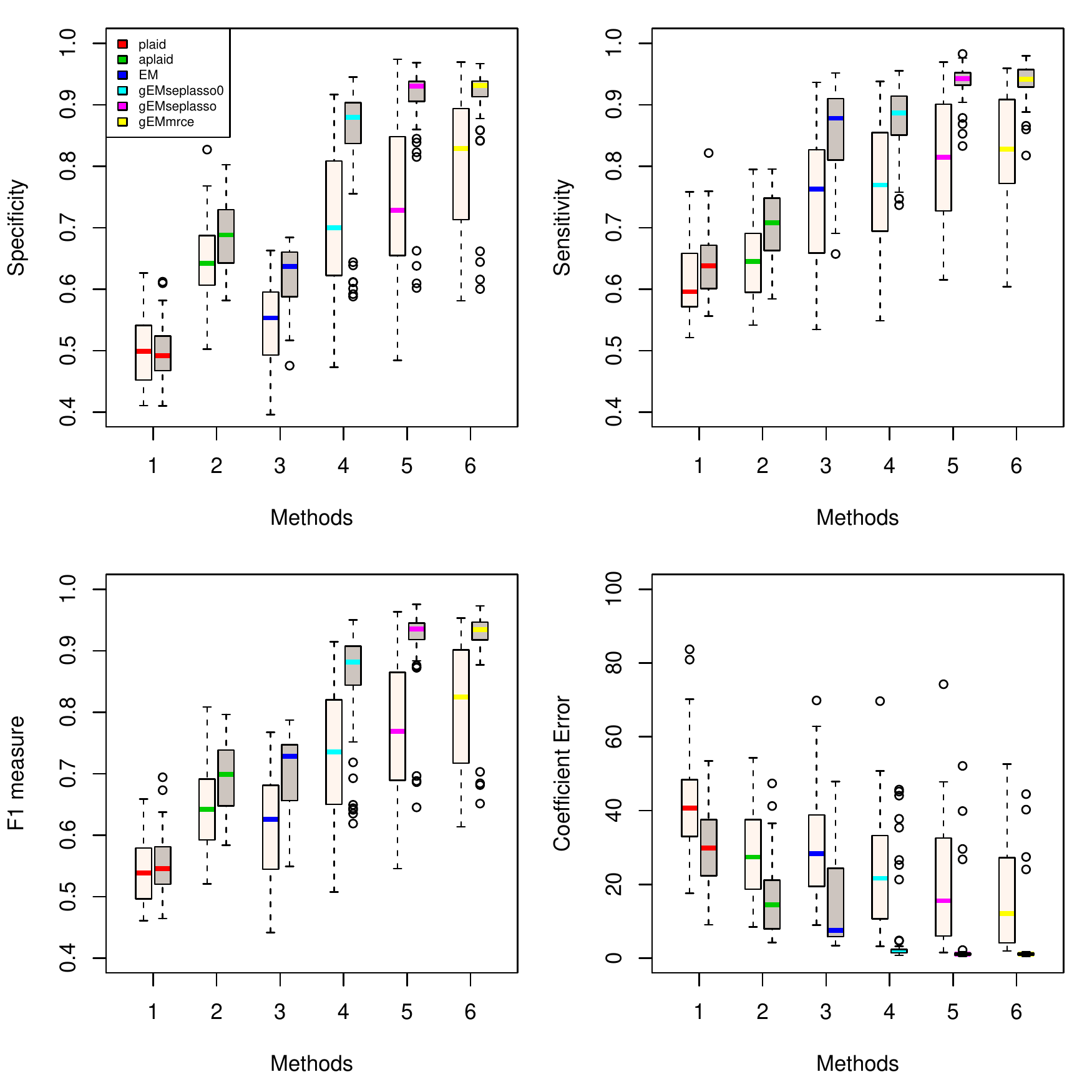}
\caption{Simulation results of scenario 2.}
\label{fig:case2}
\end{figure}

\begin{figure}
\centering
\includegraphics[scale=0.38]{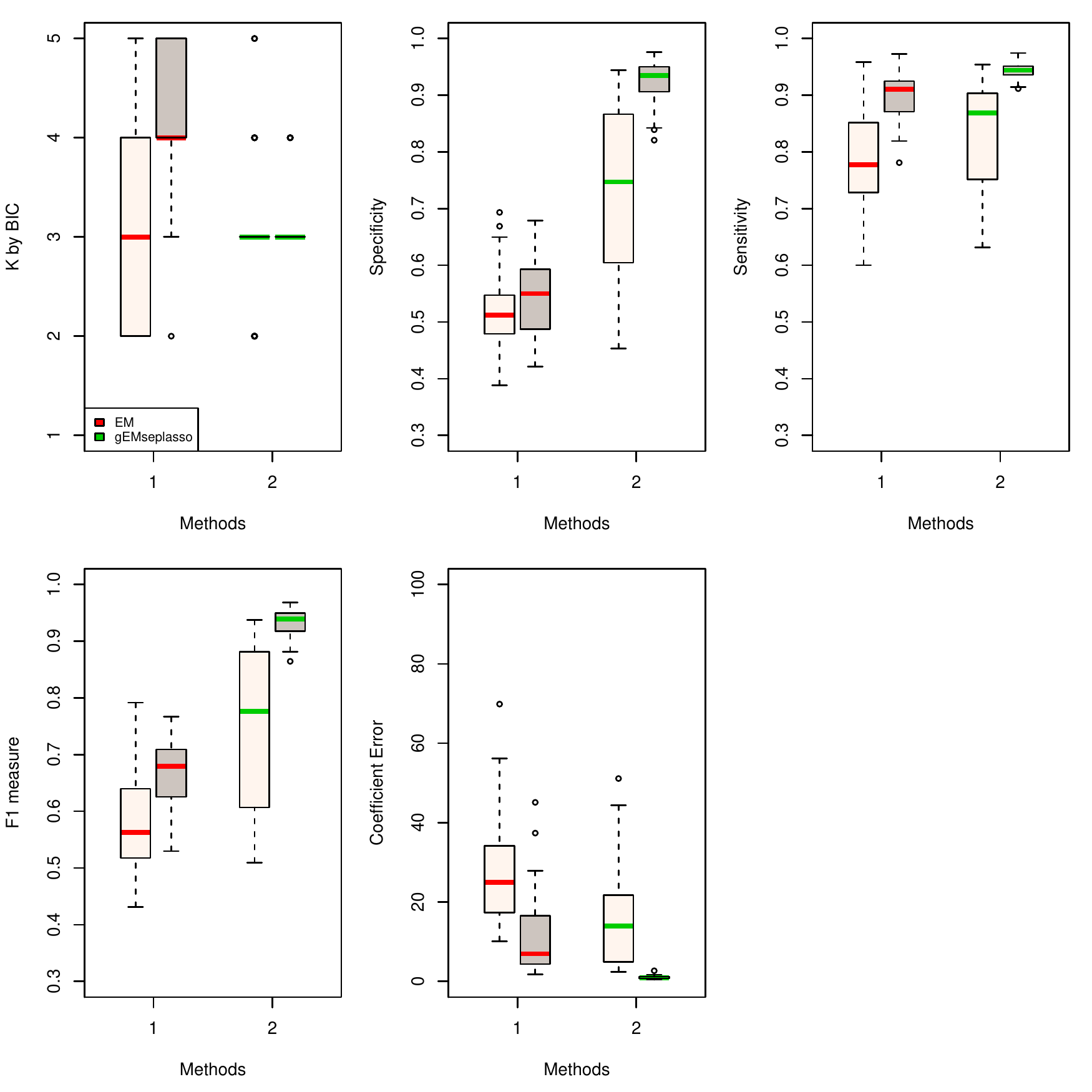}
\caption{Simulation results of BIC in selection of $K$.}
\label{fig:ksimu}
\end{figure}

\begin{figure}
\centering
\includegraphics[scale=0.38]{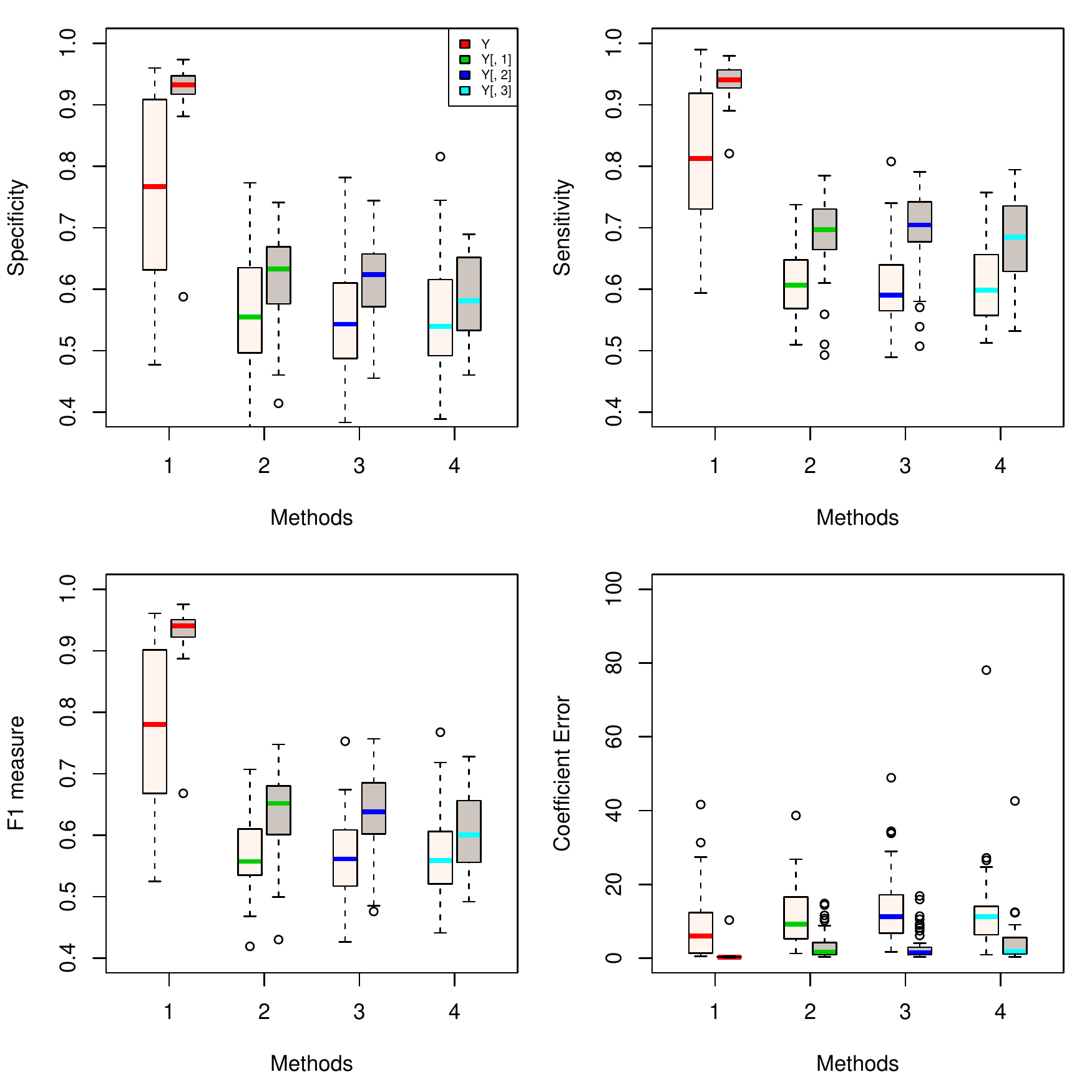}
\caption{Simulation results of using multivariate responses versus univariate responses.}
\label{fig:unisimu}
\end{figure}

\begin{figure}
\centering
\includegraphics[width=8.7cm]{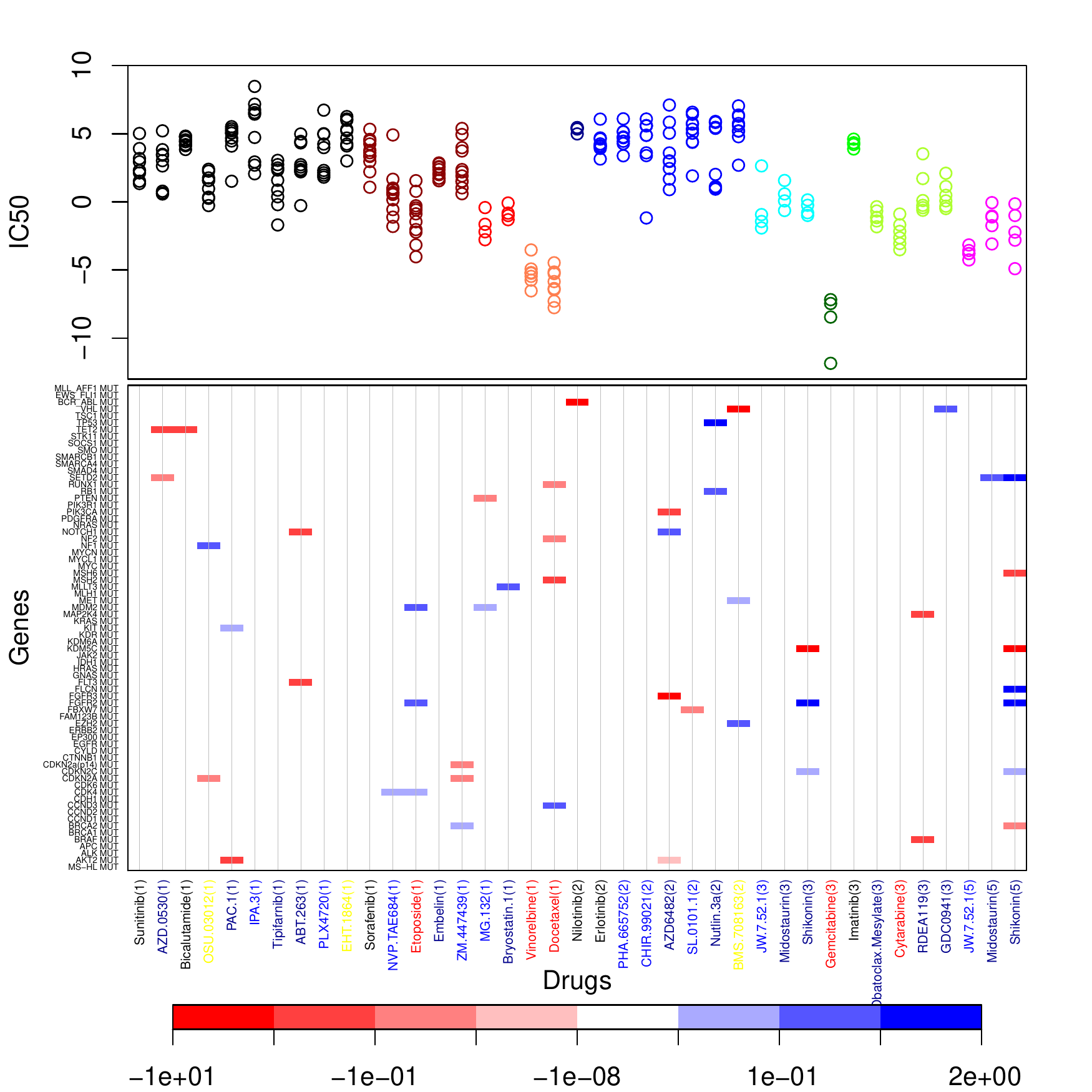}
\caption{The extracted $IC_{50}$ values of soft tissue cancers (upper panel) and a subset of the coefficient estimates corresponding to the mutations of 71 cancer genes (bottom panel).  Drug names in red indicate type ``chemo", black indicate type ``clinical", blue indicate type ``experimental", darkblue indicate type ``in clinical development" and yellow indicate drug type not available.}
\label{fig:st}
\end{figure}

\begin{figure}
\centering
\includegraphics[width=11.4cm, height=10cm, trim=85 512 150 0, clip=true]{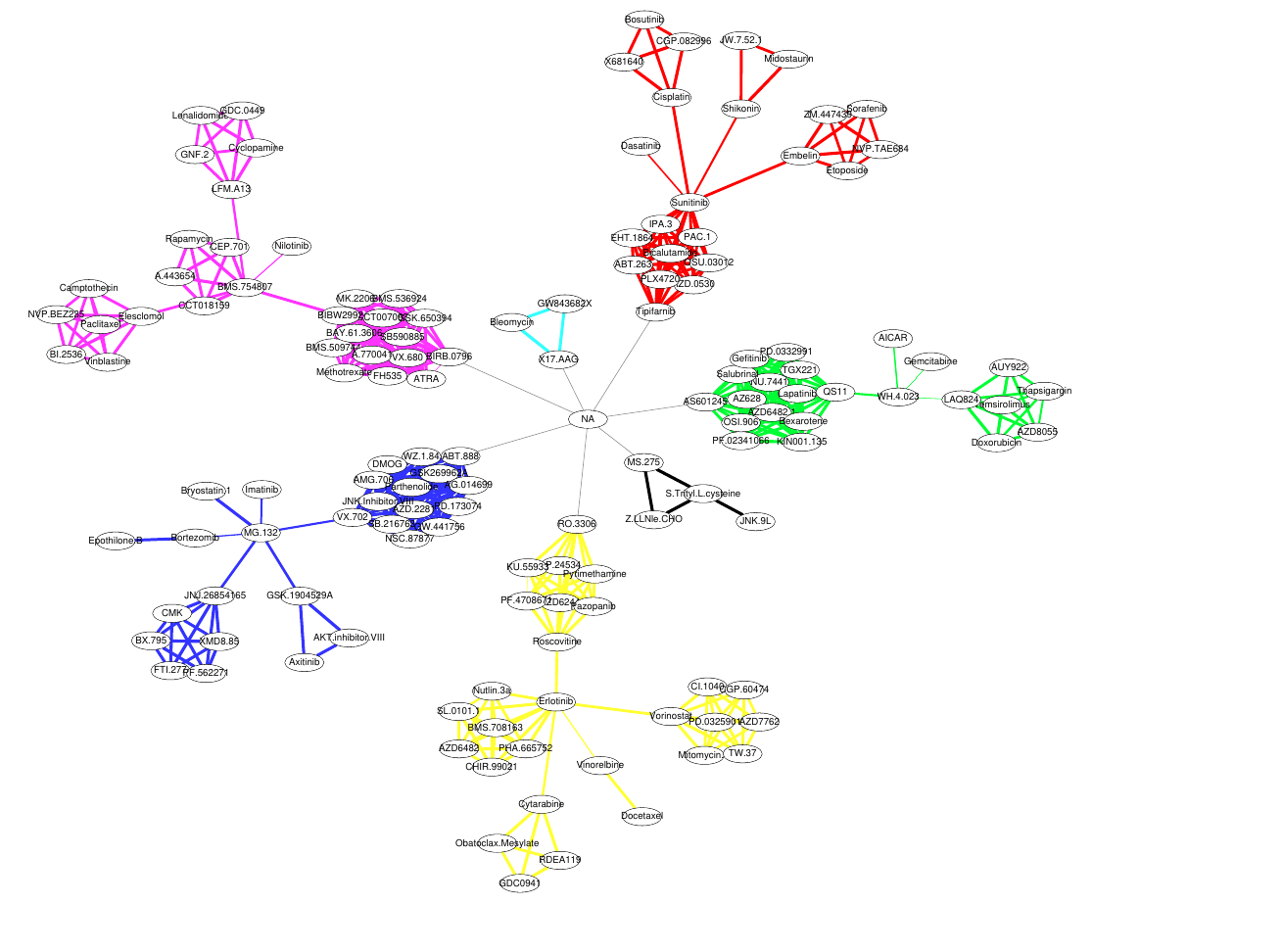}
\caption{Groups of drugs generated from the APC algorithm. Each node implies a type of drug with the drug name on it.  Drugs connected with the same color of strokes belong to a first-level drug group, and drugs connected within a closure belong to a second-level drug group.}
\label{fig:interm}
\end{figure}

\begin{figure}
\centering
\includegraphics[width=11.4cm, height=10cm]{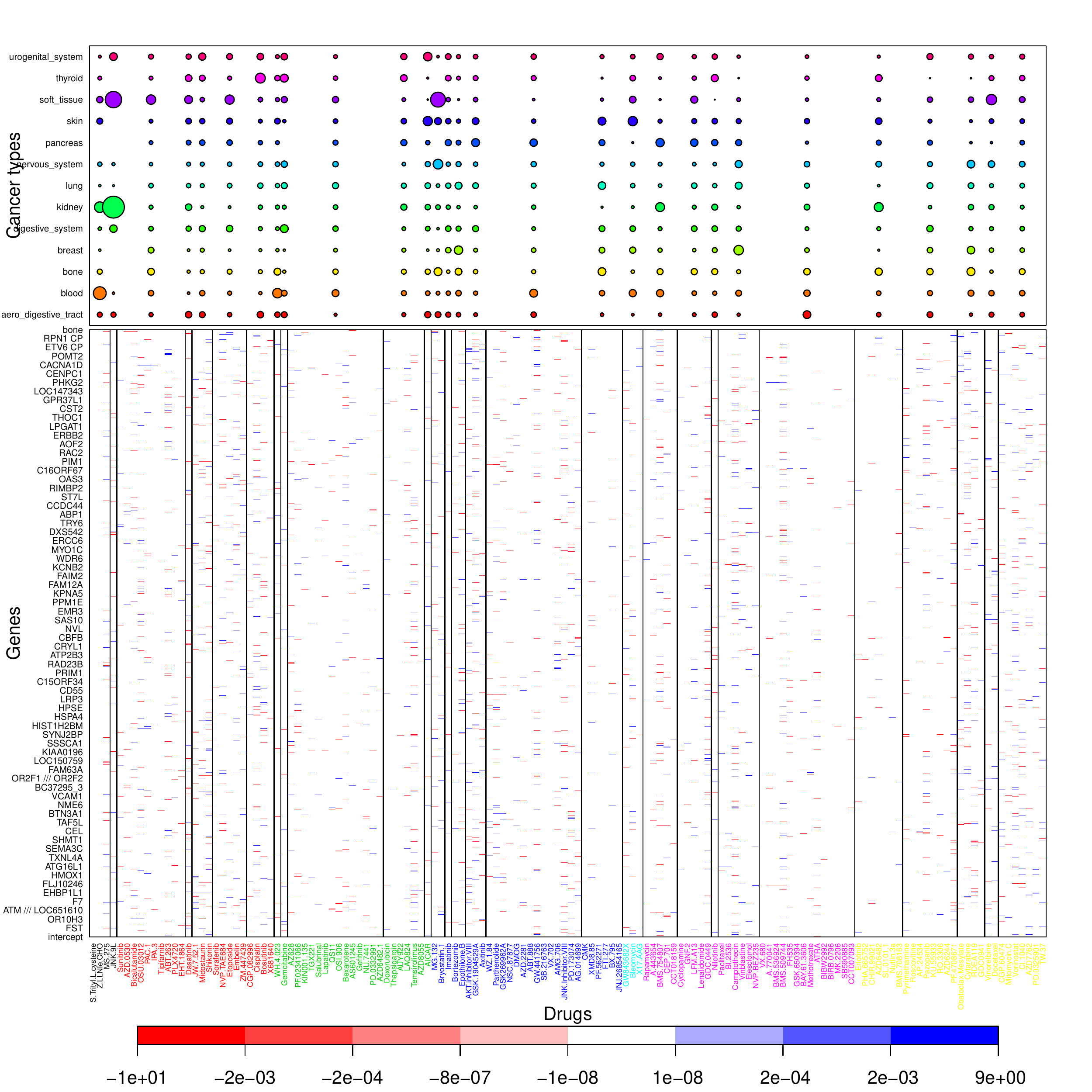}
\caption{Estimation results of cluster 1 by fitting the data of each second-level drug group in a generalized FMMR model.  Bottom panel plots the coefficient estimates, and upper panel describes the cell line compositions of cluster 1 by cancer types. }
\label{fig:s1}
\end{figure}

\begin{figure}
\centering
\includegraphics[width=11.4cm, height=10cm]{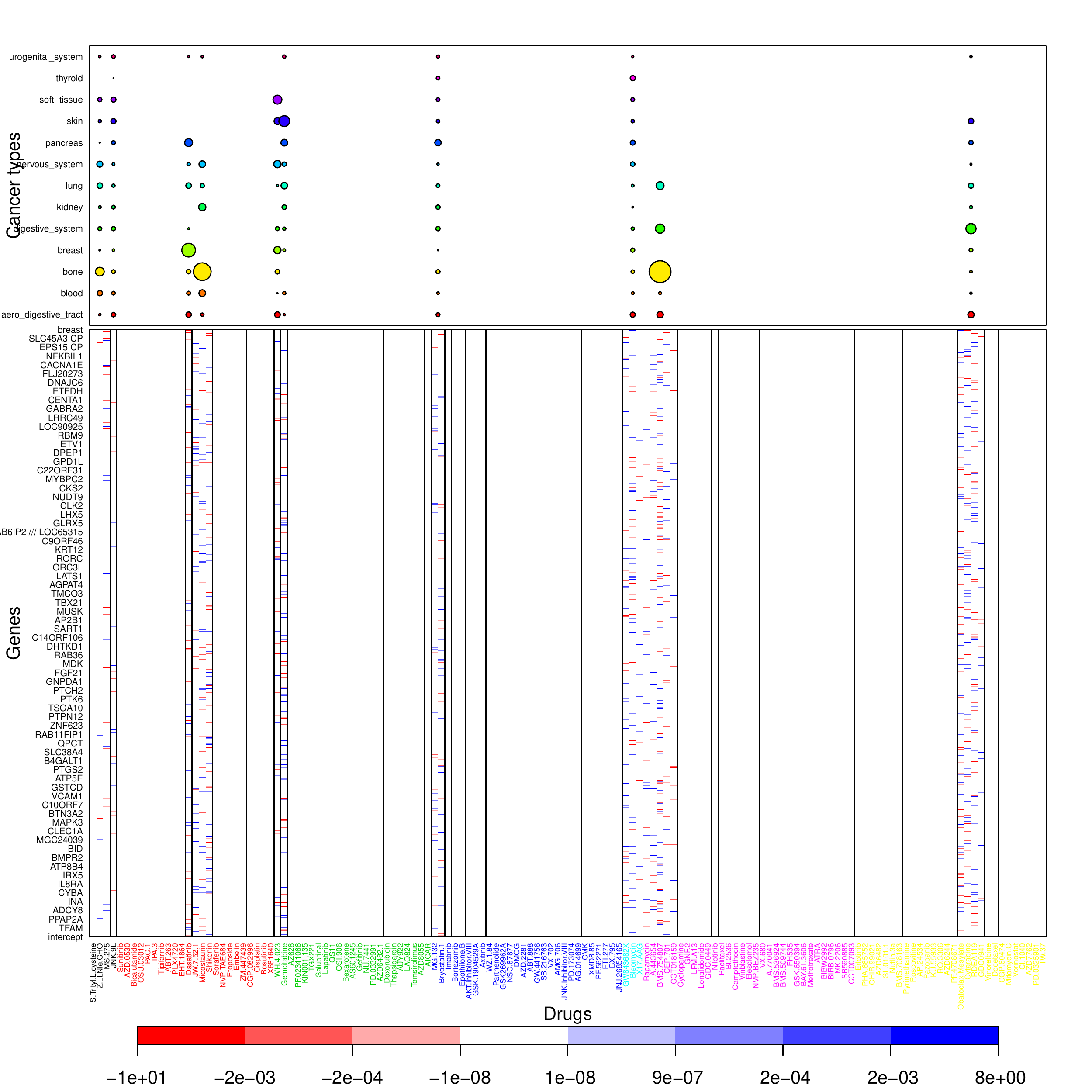}
\caption{Estimation results of cluster 12 by fitting the data of each second-level drug group in a generalized FMMR model.}
\label{fig:s4}
\end{figure}

\begin{figure}
\centering
\includegraphics[width=11.4cm, height=10cm]{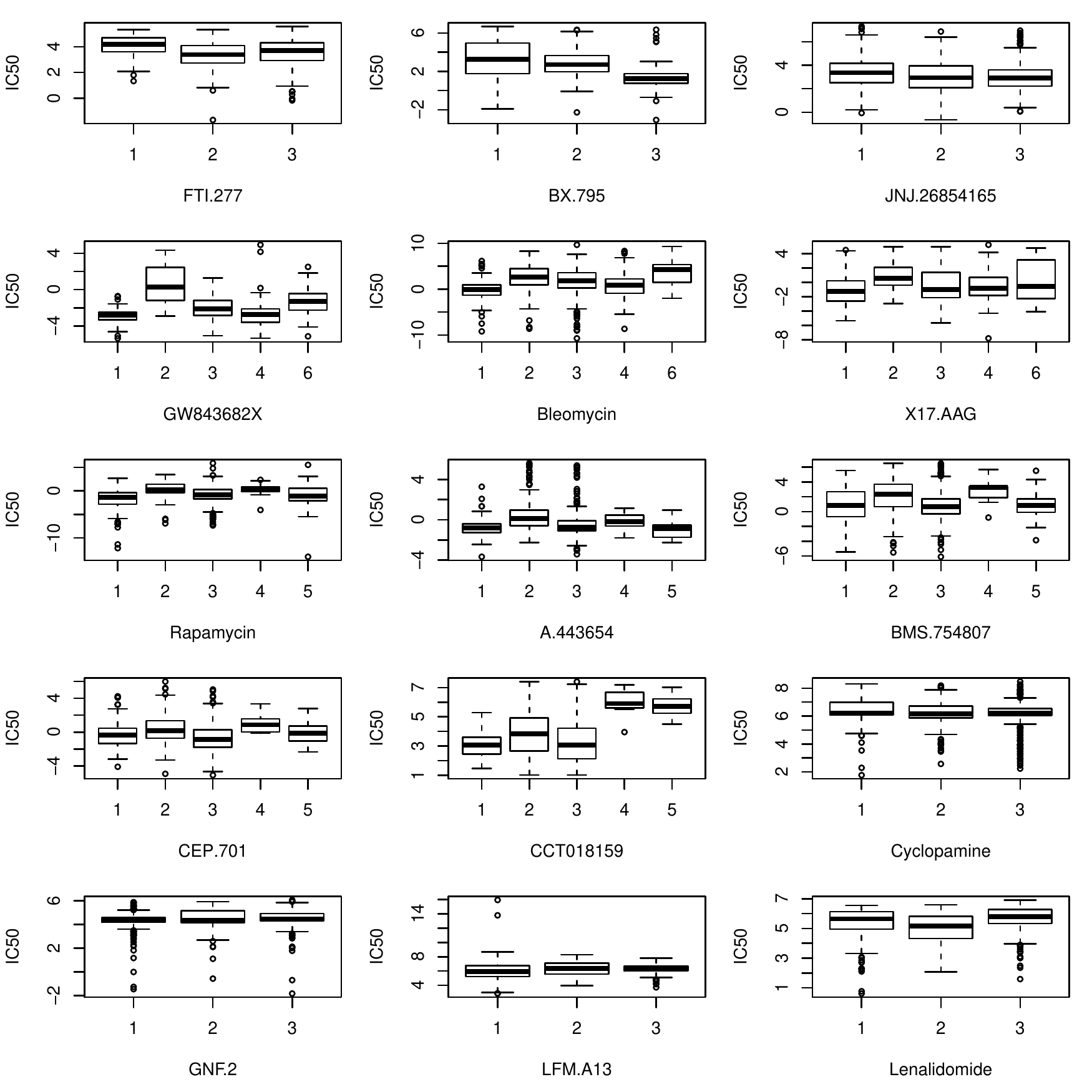}
\caption{Cluster-wise $IC_{50}$ values of each drug.}
\label{fig:ic}
\end{figure}

\end{document}